\pdfoutput=1
\documentclass[%
 reprint,
superscriptaddress,
 amsmath,amssymb,
 aps,
floatfix,
]{revtex4-2}
\usepackage{graphicx}% Include figure files
\usepackage{dcolumn}% Align table columns on decimal point
\usepackage{bm}% bold math
\usepackage{amsthm}
\usepackage{commath}
\usepackage{mathtools}
\usepackage{graphicx}
\usepackage{bigints}
\usepackage[colorlinks=true, hyperindex, breaklinks, linkcolor=blue, urlcolor=blue, citecolor=blue]{hyperref} % Physical Review style
\usepackage{amsmath,amsfonts,amssymb,color,graphicx,xcolor,physics}
\usepackage{tikz}
\usepackage{float}
\makeatletter
\let\newfloat\newfloat@ltx
\makeatother
\usepackage[linesnumbered, ruled, vlined]{algorithm2e}
\SetKwRepeat{Do}{do}{while}%
\usepackage{makecell}
\usepackage{stmaryrd}
\usepackage{tabularx}

\usepackage{tikz} 
\usetikzlibrary{quantikz}
\usetikzlibrary{arrows,shadows} 
\usepackage{multirow}
\DeclarePairedDelimiter\floor{\lfloor}{\rfloor}

\usepackage{lineno}
\newcommand{\spanv}[1]{\text{Span}\lbkt{#1}}
\newcommand{\bigO}[1]{\mathcal{O}\sbkt{#1}}
\newcommand{\im}[1]{\mathrm{im}\ #1}

\newcommand{\mc}{\mathcal}

\newcommand{\mbb}{\mathbb}
\newcommand{\mr}{\mathrm}
\newcommand{\sbkt}[1]{\left(#1\right)}
\newcommand{\mbkt}[1]{\left[#1\right]}
\newcommand{\lbkt}[1]{\left\{#1\right\}}

\newtheorem{theorem}{Theorem}
\newtheorem{lemma}{Lemma}
\newtheorem{definition}{Definition}
\newtheorem{corollary}{Corollary}
\newtheorem*{theorem*}{Theorem}
\newtheorem*{corollary*}{Corollary}
\newtheorem{example}{\textbf{Example}}
\newtheorem{remark}{Remark}

\newcommand\nonpfrate[1]{\gamma_{X, Y}}
\makeatletter
\newcommand*{\rom}[1]{\expandafter\@slowromancap\romannumeral #1@}
\makeatother
\let\oldproof\proof
\renewcommand{\proof}{\oldproof}

\def\algbackskip{\hskip-\ALG@thistlm}

\begin{document}

\preprint{APS/123-QED}

\title{High-Rate Surgery: towards constant-overhead logical operations}
\author{Guo Zheng} 
\email{guozheng@uchicago.edu}
\affiliation{Pritzker School of Molecular Engineering, The University of Chicago, Chicago 60637, USA}

\author{Liang Jiang}
\email{liangjiang@uchicago.edu}
\affiliation{Pritzker School of Molecular Engineering, The University of Chicago, Chicago 60637, USA}
\author{Qian Xu}
\email{qianxu@caltech.edu}
\affiliation{Institute for Quantum Information and Matter, Caltech, Pasadena, CA 91125, USA}
\affiliation{Walter Burke Institute for Theoretical Physics, Caltech, Pasadena, CA 91125, USA}
% \maketitle

\date{\today}% It is always \today, today,

\begin{abstract}
Scalable quantum computation requires not only quantum codes with low memory overhead but also encoded operations with low space–time overhead. High-rate quantum low-density parity-check (qLDPC) codes address the former by achieving a high information-encoding rate, yet existing methods for implementing logical operations often suffer from a low \emph{information-processing rate}, leading to substantial space–time costs.
Here, we introduce \emph{high-rate surgery}, a general scheme that can perform extensive, addressable logical Pauli-product measurements in parallel on arbitrary qLDPC codes using a shared ancilla system, attaining nearly constant space–time overhead. 
We develop both algebraic and randomized ancilla constructions and demonstrate, using the $[[144, 12, 12]]$ Gross code and new instances of qLDPC codes (e.g., $[[1125, 245, \leq 10]]$) with encoding rate up to $25\%$, that up to \emph{hundreds} of randomly sampled logical measurements can be executed simultaneously with a total space-time overhead within a factor of two of that of memory experiments.
Our results address a major bottleneck for performing complex, addressable logical operations on qLDPC codes in practice, advancing the prospect of scalable, constant-overhead fault-tolerant quantum computation.

\end{abstract}

\maketitle

%\tableofcontents

\section{Introduction}

\begin{figure*}
    \centering
    \includegraphics[width = 0.9 \textwidth]{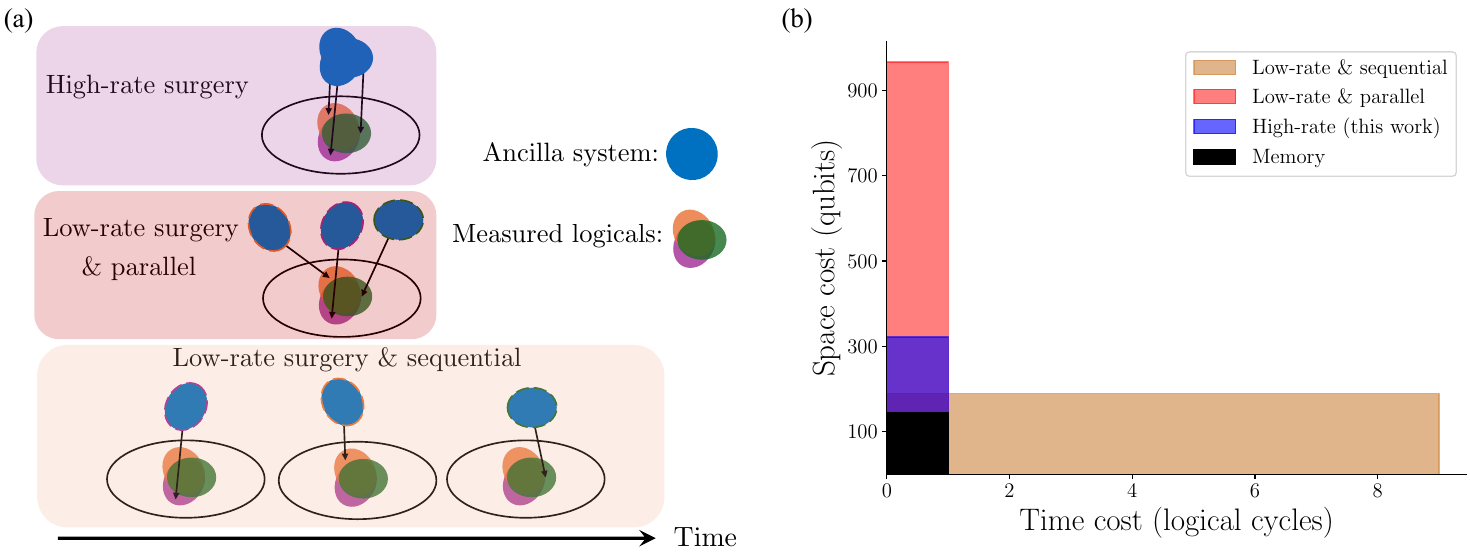}
     \caption{\textbf{Illustration of high-rate versus low-rate surgery schemes}. (a) Comparison of low-rate (sequential~\cite{williamson2024lowoverheadfaulttolerantquantumcomputation,Cohen2022LDPCSurgery} and parallel~\cite{Cowtan2025ParallelSurgery,Guo2025TimeEfficient}) and high-rate surgery scheme (this work). The horizontal axis represent time axis. 
     % and each figure represent a time interval of one logical cycle.
     Blue circles denote ancilla systems that couple to multiple target logical operators, whose supports are shown by the green, orange, and pink circles, possibly with overlaps.
     (b) The space (vertical axis) and time (horizontal) costs of surgery schemes on Gross code~\cite{Bravyi2024Nature}, $[[144, 12, 12]]$, for measuring nine randomly chosen logicals, compared against memory. Space costs are measured in unit of physical qubits. Time costs are in units of logical cycles.
     }
    \label{fig:illustrate}
\end{figure*}

Quantum error correction (QEC) is a cornerstone of large-scale, fault-tolerant quantum computation~\cite{Shor1995QEC,Gottesman1997Stabilizer,Terhal2015QECReview,Kitaev2003Topological}.
Much of the literature has focused on reducing the space-time overhead of quantum memories—that is, the cost of storing quantum information reliably over time.
However, beyond protecting idle qubits, practical quantum computers should also support logical operations with comparably low overhead~\cite{Litinski2019Game}.

Quantum low-density parity-check (qLDPC) codes~\cite{Breuckmann2021QLDPCPerspective} are quantum error-correcting codes with sparse parity-check matrices.
Though nonlocal, they are ideal for low-overhead quantum memories and form the basis of many recent experiments and proposals~\cite{wang2025demonstrationlowoverheadquantumerror, yoder2025tourgrossmodularquantum,Bravyi2024Nature,Xu_2024_NaturePhysics, andersen2025smallquantumlowparity, Pecorari_2025}.
Nevertheless, it remains open whether they retain this overhead advantage for processing quantum information.
Well-known approaches to fault-tolerant logical operations include transversal (and related) gates~\cite{Breuckmann2021QLDPCPerspective} and automorphism gates~\cite{Breuckmann2024FoldTransversal, Sayginel2025Automorphisms, eberhardt2024logicaloperatorsfoldtransversalgates, Quintavalle_2023}, but these are often tailored to specific code families and lack full addressability.
Originating in lattice surgery for topological codes~\cite{Cohen2022LDPCSurgery, Horsman_2012, Fowler_2012, landahl2014quantumcomputingcolorcodelattice, Litinski2019Game}, code surgery~\cite{Cross2024ImprovedSurgery, Guo2025TimeEfficient, cowtan2025homologyhopfalgebrasquantum, Cowtan2025ParallelSurgery, he2025extractorsqldpcarchitecturesefficient, cowtan2024ssipautomatedsurgeryquantum, Cowtan_2024, cowtan2022quditlatticesurgery, williamson2024lowoverheadfaulttolerantquantumcomputation, Ide_2025,swaroop2025universaladaptersquantumldpc} is one of the few schemes available for arbitrary qLDPC codes.
It enables universal computation via Pauli-product measurements (PPMs)~\cite{Litinski2019Game, PhysRevX.6.021043} combined with magic-state injection~\cite{PhysRevA.71.022316}.

Inspired by techniques from quantum weight reduction~\cite{hastings2016weightreductionquantumcodes, hastings2023quantumweightreduction, Sabo_2024}, there has been tremendous progress in reducing the overhead of code surgery.
Nevertheless, most work has focused on measuring a single logical operator.
Such an assumption fundamentally limits the scalability of surgery schemes for dense computation.
From a space-time perspective, these protocols adopt ancilla systems that are inherently low rate, which limits the information-processing density—how many logical operations are performed per physical operation.
Asymptotically, this can negate the benefits of high-rate quantum memory.

In this work, we propose a framework for high-rate surgery that significantly lowers overhead; see Fig.~\ref{fig:illustrate}(a) for a comparative illustration.
The key idea is to employ high-rate ancilla patches for information processing.
We formalize the logical action of such surgeries in a homological framework~\cite{Ide_2025, cowtan2025homologyhopfalgebrasquantum, Xu2025Homomorphic,Huang2023Homomorphic} and prove sufficient conditions for fault tolerance.
We further give explicit constructions of high-rate ancillas that are distance preserving.
When the memory possesses suitable structure, our ancilla construction achieves amortized constant overhead for up to $\bigO{k}$ logical operations.
Hypergraph-product (HGP) codes~\cite{Tillich2014HGP} are examples of such structured codes.

In the absence of code structure, we demonstrate the power of high-rate surgery with a randomized ancilla-construction algorithm that is high rate and introduces only minimal increases in connectivity.
We benchmark its performance on Gross codes~\cite{yoder2025tourgrossmodularquantum,Bravyi2024Nature} and spatially coupled (SC) codes~\cite{Yang2025SCQLDPC}, using a randomly sampled logical basis and measurements.
For the latter, we present new code families with encoding rates up to $25\%$ and block lengths up to $1125$ qubits.
On these practical instances, high-rate surgery achieves nearly constant space-time overhead across a range of code distances and sizes.
The constant factor is modest—within a factor of two of the memory baseline (the space-time cost of one syndrome-extraction cycle on an idling system).
Figure~\ref{fig:illustrate}(b) compares space-time costs across schemes.
In contrast to low-rate approaches, high-rate surgery has spatial and temporal footprints comparable to memory.
Our results suggest that logical operations can be made nearly as cheap as memory, marking a step toward constant-overhead, fault-tolerant quantum computation.

\section{Results}

The space overhead of a quantum memory with code parameters $[[n,k,d]]$ is captured by its \emph{encoding rate}, $r := k/n$.
High-rate qLDPC codes cap this overhead by achieving a \emph{constant encoding rate}, $r=\Theta(1)$, whereas in low-rate codes $r \to 0$ as $n \to \infty$.

When computing on encoded logical qubits, one must also account for space–time trade-offs. Thus, the total cost of computation is best quantified by the \emph{space--time overhead}.
For a circuit $C$ with width $W(C)$ and depth $D(C)$, define its space--time cost as $S(C) := W(C) \times D(C)$.
An encoded circuit $\overline{C}$ implementing the same logical operations has a corresponding cost $S(\overline{C})$.
The \emph{space--time overhead}, $S(\overline{C})/S(C)$, quantifies the additional resources—extra physical qubits for encoding redundancy and additional circuit depth for fault-tolerant logical operations—needed to perform the computation reliably in the presence of noise.

In typical fault-tolerant protocols, each elementary logical operation consists of a \emph{logical cycle}, comprising multiple rounds of syndrome extraction to reliably infer syndromes and remove entropy from the system.
Although complementary works explore ways to shorten the logical cycle~\cite{Campbell_2019, Quintavalle_2023, zhou2025lowoverheadtransversalfaulttolerance, PhysRevLett.133.240602, baspin2025fastsurgeryquantumldpc, xu2025batchedhighratelogicaloperations}, we assume that for a distance-$d$ code each logical cycle takes time $t_L=\Theta(d)$ unless stated otherwise.
We then define the \emph{normalized space--time overhead} as
\begin{equation}\label{eq:spacetime_overhead}
    \alpha := \frac{1}{t_L}\,\frac{S(\overline{C})}{S(C)},
\end{equation}
which measures overhead in units of (physical qubits) $\times$ (logical cycles).

As a baseline, an idling memory operation—one code block undergoing a single logical cycle—has normalized overhead equal to the inverse of the encoding rate, $\alpha=\Theta(1/r)$.
For constant-rate qLDPC codes, this baseline overhead is $\Theta\sbkt{1}$.
Our goal is to design logical operations whose overhead approaches this memory limit.\\

\begin{figure}[t!]
    \centering
    \includegraphics[width = 0.5 \textwidth]{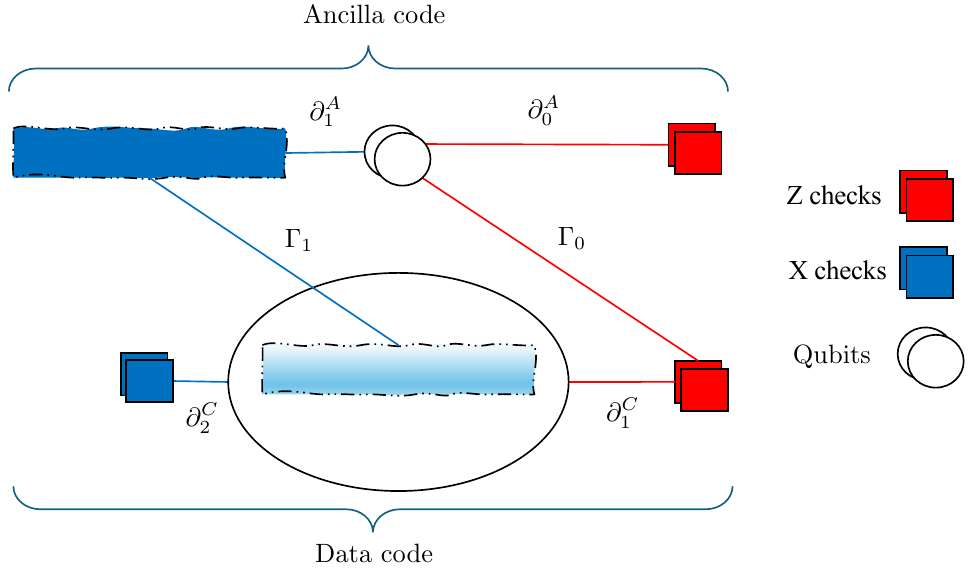}
     \caption{\textbf{Tanner graph illustration of a general surgery scheme}. Checks (resp. qubits) are represented by squares (resp. circles)~\cite{williamson2024lowoverheadfaulttolerantquantumcomputation}. The shown figure presents the ancilla system merged with the data system to measure X-logical operators, according to Eq.~\eqref{eq:cone_code}. 
     At the end, the ancilla qubits are transversally measured in Z-basis. The example illustrates a transversal-type $\Gamma_1$ such that ancilla X checks and the measured operators' support (indicated by the box with a color gradient) are of similar sizes.
     }
    \label{fig:illustrate_merge}
\end{figure}

\begin{table*}
    \centering
    \begin{tabular}{c c c c }
        \hline \hline
        Operation & Scheme & Space overhead & Time overhead\\
         \hline 
         Logical measurement & Low-rate surgery, sequantial & $\Theta\sbkt{1}$ & $\Theta\sbkt{k}$ \\
        Logical measurement & Low-rate surgery, parallel & $\tilde{\Theta}\sbkt{d}^\dagger$ & $\Theta\sbkt{1}$\\
        Logical measurement &High-rate surgery & $\Theta\sbkt{1}^\ast$& $\Theta\sbkt{1}$\\
        Memory & Syndrome extraction & $\Theta\sbkt{1}$& $\Theta\sbkt{1}$\\
         \hline \hline 
    \end{tabular}
    \caption{\textbf{Summary of asymptotical spacetime overheads of various surgery schemes}, compared against memory, for a $[[n, k, d]]$ qLDPC code. Space overhead is normalized by unencoded circuit cost, and time overhead is normalized by the length of a logical cycle, $\Theta(d)$. For simplicity, assume $\Theta\sbkt{k}$ logicals are measured, and the memory is constant encoding rate. $\Theta\sbkt{1}^\ast$ assumes Theorem~\ref{theorem:algebraic_surgery}. $\tilde{\Theta}\sbkt{d}^\dagger$ assumes minimum-weight basis is available.}
    \label{tab:scheme_overview}
\end{table*}

\noindent
\textbf{Code surgery framework.}
In what follows, we consider logical circuits built from logical Pauli-product measurements (PPMs), which, together with magic-state injection, realize universal Pauli-based computation~\cite{Litinski2019Game, PhysRevA.71.022316}. For concreteness, we focus on $X$-type PPMs for a Calderbank–Shor–Steane (CSS) code~\cite{PhysRevA.54.1098, PhysRevLett.77.793}; measurements of $Z$-type logicals follow symmetrically. Simultaneous measurement of mixed-type Paulis is discussed in Methods.

CSS codes are quantum codes whose parity-check matrices admit a generating basis split into $X$- and $Z$-type checks, $H_X$ and $H_Z$, respectively. Their commutation relations can be expressed with a $2$-chain complex $\mc{C}$:
\begin{equation}
        \mc{C}: \begin{tikzcd}[column sep=large, row sep=large]
        {C_2} & {C_1} & {C_0}
        \arrow["{\partial^C_2}", from=1-1, to=1-2]
        \arrow["{\partial^C_1}", from=1-2, to=1-3]
\end{tikzcd},
\end{equation}
with boundary maps $\partial_2^C = H_X^T$ and $\partial_1^C = H_Z$ satisfying $\partial^C_1 \partial^C_2 = 0 \bmod 2$. The spaces $C_l$, for $l=2,1,0$, are $\mbb{F}_2$-vector spaces whose bases are associated with $X$-checks, qubits, and $Z$-checks, respectively.

To perform a surgery logical operation on a data code $\mc{C}$, we attach an ancilla system
\begin{equation}
    \mc{A}: \begin{tikzcd}[column sep=large, row sep=large]
        {A_1} & {A_0} & {A_{-1}}
        \arrow["{\partial^A_1}", from=1-1, to=1-2]
        \arrow["{\partial^A_0}", from=1-2, to=1-3]
\end{tikzcd},
\end{equation}
where $A_l$ with $l=1,0,-1$ are associated with the $X$ checks, qubits, and $Z$ checks of $\mc{A}$, respectively.
We attach $\mc{A}$ to $\mc{C}$ and aim to extract certain logical information of $\mc{C}$ through measurements of the $X$ checks $A_1$ on $\mc{A}$.
The “attaching” procedure, and the resulting joint data-ancilla code, are described by the following complex~\cite{Ide_2025}:
\begin{definition}[Code surgery]
\label{def:code_surgery}
For a data code $\mc{C}$ and an ancilla system $\mc{A}$, code surgery is uniquely defined by the complex
    \begin{eqnarray}\label{eq:cone_code}
        \begin{tikzcd}[column sep=large, row sep=large]
        {A_1} & {A_0} & {A_{-1}} \\
        {C_2} & {C_1} & {C_0}
        \arrow["{\partial^A_1}", from=1-1, to=1-2]
        \arrow["{\partial^A_0}", from=1-2, to=1-3]
        \arrow["{\partial^C_2}", from=2-1, to=2-2]
        \arrow["{\partial^C_1}", from=2-2, to=2-3]
        \arrow["{\Gamma_1}", from=1-1, to=2-2]
        \arrow["{\Gamma_0}", from=1-2, to=2-3]
\end{tikzcd}
\end{eqnarray}
where the diagram commutes, i.e., $\Gamma_0 \partial^A_1 = \partial^C_1 \Gamma_1$.
Here $\Gamma_1$ (respectively $\Gamma_0$) describes how the $X$ checks $A_1$ of $\mc{A}$ (respectively the $Z$ checks $C_0$ of $\mc{C}$) are connected to the qubits $C_1$ of $\mc{C}$ (respectively $A_0$ of $\mc{A}$).
\end{definition}
See also Figure~\ref{fig:illustrate_merge} for a Tanner-graph illustration of the merged code.
After one logical cycle on the merged code in Eq.~\eqref{eq:cone_code}, during which the logical information of $\mc{C}$ is reliably extracted, we split the two codes and return to the codespace of $\mc{C}$ by transversally measuring out the ancilla $\mc{A}$.
While we omit the details of the merge-and-split procedure here, it is crucial to initialize and measure $\mc{A}$ in the correct bases to ensure fault tolerance.
The required procedure has been established in prior works~\cite{Horsman_2012,Vuillot_2019,Cohen2022LDPCSurgery,williamson2024lowoverheadfaulttolerantquantumcomputation}, and a complete description is provided in Methods.
We also refer interested readers to Ref.~\cite{he2025extractorsqldpcarchitecturesefficient} for a review of the history and development of surgery techniques.

The essential idea of the surgery procedure is weight reduction: global logical $X$ operators in the data code are decomposed into products of local stabilizers from the $X$ checks $A_1$ of the merged code, whose values are reliably inferred during a logical cycle.
% At the final stage, the merged code returns to the data code through a transversal measurement on the ancilla qubit.
This merge-and-split weight-reduction process is also known in the condensed-matter community as gauging and ungauging symmetries~\cite{PhysRevB.94.155128, kubica2018ungaugingquantumerrorcorrectingcodes}.
More concretely, the logical information extracted by surgery is given by the product of stabilizers in $A_1$ with no support on the ancilla qubits $A_0$~\cite{Ide_2025}:
\begin{lemma}[Information extracted through code surgery]\label{lem:log_action}
    For a surgery procedure in Def.~\ref{def:code_surgery}, the space of all measured operators of $\mc{C}$ is
    \begin{eqnarray}
            M = \Gamma_1\sbkt{\ker{\partial_1^A}} \subseteq S_X^C \oplus L_X^C, \label{eq:log_measurement} 
    \end{eqnarray}
    where $S_X^C$ and $L_X^C$ denote the spaces of the $X$ checks and the $X$ logical operators of $\mc{C}$, respectively.
\end{lemma}

As implied by Lemma~\ref{lem:log_action}, the number of bits of logical information extracted is $\dim M$, which is closely tied to the kernel dimension of $\partial^A_1$.
Formally, the information-extraction efficiency is measured by the information-extraction rate:
\begin{definition}[Information extraction rate (IER)]
Consider a surgery procedure with space of measured operators $M$ in Lemma~\ref{lem:log_action} using an ancilla system $\mc{A}$. The information-extraction rate is
    \begin{equation}
        r_M := \dim\sbkt{M}/\abs{\mc{A}},
    \end{equation}
    where $\abs{\mc{A}} = \sum_{i=-1}^1 \dim A_i$ is the total size of the ancilla system (including qubits and checks).
\end{definition}

We now show that, just as the encoding rate $r$ of a code governs the space overhead of a quantum memory, the IER $r_M$ controls the space–time overhead of a surgery-based quantum processor.
Consider a shallow yet dense circuit executing $t=\Theta\sbkt{k}$ PPMs on $k$ qubits, where each qubit participates in $\mc{O}(1)$ PPMs. This implies each PPM involves only $\mc{O}(1)$ qubits.
Such a pattern corresponds to parallel computation on all $\Theta\sbkt{k}$ logical qubits, for example injecting magic states on all logical qubits in parallel, which in an unencoded setting can be executed in depth $\mc{O}(1)$.
The spacetime overhead associated with surgery measurements using an ancilla system $\mc{A}$ with fixed IER $r_M$ is~\footnote{The spacetime overhead of the full scheme is $
    \alpha = \frac{|\mathcal{C}| + |\mathcal{A}|}{k + t}\times \frac{t}{r_M |\mathcal{A}|}= \bigO{1/r_M}$.}
\begin{equation}
    \alpha
    = \underbrace{\frac{|\mc{A}|}{k}}_{\text{space}}
      \times
      \underbrace{\frac{t}{r_M\,|\mc{A}|}}_{\text{time}}
    = \Theta\!\left(\frac{1}{r_M}\right),
    \label{eq:surgery_IER}
\end{equation}
where the space factor $\tfrac{|\mc{A}|}{k}$ compares the total ancilla size to the space cost of the unencoded circuit, and the time factor reflects measuring $t$ PPMs in $\tfrac{t}{r_M|\mc{A}|}$ logical cycles, each cycle completing $r_M|\mc{A}|$ PPMs using $\mc{A}$ once.

The above analysis admits space–time trade-offs.
In a \emph{sequential} surgery scheme with $r_M|\mc{A}|\ll t$, the space cost can be small because one can take $|\mc{A}|=\mc{O}(k)$, but the time cost is large due to sequential execution (see the yellow box in Fig.~\ref{fig:illustrate}(a)).
In a \emph{parallel} surgery scheme with $r_M|\mc{A}|=\Theta(t)$, the PPMs complete in parallel with small time cost, yet the space cost can become large with $|\mc{A}|=\Theta(t)/r_M$ when $r_M$ is low (see the pink box in Fig.~\ref{fig:illustrate}(a)).
In all cases, the overall space–time overhead is characterized by the IER $r_M$.

Most prior surgery schemes are low rate: they operate under the restriction that the ancilla system consists of one or more disjoint subsystems each with $\dim\ker\partial_1^A=1$ and $|\mc{A}|=\tilde{\Omega}\sbkt{d}$,
which yields IER $r_M=\tilde{\mc{O}}(1/d)$.
Consequently, they incur large overhead in either space or time, depending on whether the implementation is parallel or sequential (see Fig.~\ref{fig:illustrate}(b)).

% More concretely, low-rate surgery schemes can be divided into sequential~\cite{williamson2024lowoverheadfaulttolerantquantumcomputation} and parallel~\cite{Cowtan2025ParallelSurgery} categories. The sequential scheme trades time for space by having $r_M |\mathcal{A}| = \mathcal{O}(1)$ and invoking a time overhead of $\Theta\sbkt{k}$. On the other hand, the parallel scheme trades space for time with $r_M |\mathcal{A}| = \Theta\sbkt{k}$, leading to a space overhead of $\Omega\sbkt{d}$~\footnote{Here, for simplicity, we have neglected a polylogarithm overhead for the parallel scheme. On a high level, this is the cost of retaining code sparsity while trading space for time, which is achieved through a tree structure.} Thus, it is clear that the two schemes are connected as spacetime tradeoffs. More generally, to measure $t$ logicals, the sequential and parallel schemes each have $\Theta\sbkt{t}$ and $\Omega\sbkt{\frac{t}{k}d}$ overhead, with an unavoidable linear scaling in $t$ for both schemes due to a low IER. Moreover, the hidden constant factors are often significant, which is evident in Figure.~\ref{fig:illustrate}(b). See Table~\ref{tab:scheme_overview} for a concise overview.

Therefore, in this work we introduce \emph{high-rate} surgery schemes.
In particular, we design ancilla systems that extract $\dim(M)=\Theta(k)$ bits of information collectively, using a joint ancilla of bounded size $|\mc{A}|=\mc{O}(n)$ (see the purple box in Fig.~\ref{fig:illustrate}), which for constant-rate codes implies a constant IER $r_M=\Theta(1)$ and thus low space-time overhead by the analysis above (see Fig.~\ref{fig:illustrate}(b)). We note that in the numerical analysis of the space-time overhead in Fig.~\ref{fig:illustrate}(b) and Fig.~\ref{fig:random_algo_results}, we also include the memory overhead in Eq.~\eqref{eq:surgery_IER}, which significantly increases the overall cost of sequential schemes but only marginally impact the parallel schemes.

A code-surgery procedure is fault-tolerant if (i) the boundary and chain maps in Eq.~\eqref{eq:cone_code} are LDPC and (ii) the phenomenological distance of the procedure is $\Omega(d)$.
The first condition guarantees low-degree circuit connectivity, limiting error propagation under circuit-level noise; the second ensures sufficient error protection throughout surgery.
Together these conditions imply the existence of thresholds for surgery-based computation~\cite{gottesman2013fault}.

Past work shows that the phenomenological distance is preserved if and only if the merged-code distance is preserved~\cite{Cross2024ImprovedSurgery,williamson2024lowoverheadfaulttolerantquantumcomputation}.
Although the merged code in Eq.~\eqref{eq:cone_code} is a stabilizer code with commuting checks, we treat it as a subsystem code by identifying the new logical operators in the ancilla—logical $Z$ and $X$ associated with $H_0(\mc{A})$ and $H^0(\mc{A})$—as gauge operators, since these degrees of freedom do not carry logical information.
We refer to the dressed distance of this subsystem code as the merged-code distance $d_{\mr{merge}}$ for brevity.

The low-rate surgery schemes satisfy these two fault-tolerance conditions by forcing $\mc{A}$ to be a graph and leveraging graph-theoretic tools~\cite{Freedman_2021}.
However, high-rate schemes inevitably require $\mc{A}$ to be a hypergraph, which creates new challenges for fault tolerance.
For example, random-expander arguments no longer apply, and the gauge logicals of $\mc{A}$ cannot be promoted to stabilizers.
We therefore introduce the following lemma as a sufficient condition for preserving the merged-code distance, using the “soundness” of a linear map~\cite{Campbell_2019}.
We say a boundary $\partial$ has soundness $\rho$ up to $t$ if
$\rho \le \frac{\abs{\partial x}}{d\sbkt{x,\ker\partial}}$ for all $x$ with $\abs{\partial x}\le t$,
where the reduced distance is $d(x,V):=\min_{u\in V}\abs{x+u}$.
See the appendix for a more detailed introduction to expansion-related properties.

\begin{lemma}[Sufficient condition for preserving merged-code distance]\label{lem:merge_dist}
If $\partial_1^{A}$ has soundness $\rho$ up to $t\ge d$, and the ancilla code $\mc{A}$ has $X$-distance $d_A^X\ge d$, then the merged-code distance satisfies
\begin{eqnarray}
    d_{\mr{merge}} \ge \min\sbkt{1,\ \frac{\rho}{\abs{\Gamma_1}}}\, d,
\end{eqnarray}
where $\abs{\Gamma_1}$ denotes the maximum degree of $\Gamma_1$.
\end{lemma}

We remark that this condition is sufficient but not necessary~\cite{swaroop2025universaladaptersquantumldpc}, especially when additional structure is present in the data code or the chain maps. A more detailed discussion appears in the appendix, and we note where this condition can be relaxed when appropriate.

Next, we present high-rate, fault-tolerant surgery schemes, using algebraic constructions for structured codes and random constructions for unstructured codes. \\

\begin{figure}[t!]
    \centering
    \includegraphics[width = 0.5 \textwidth]{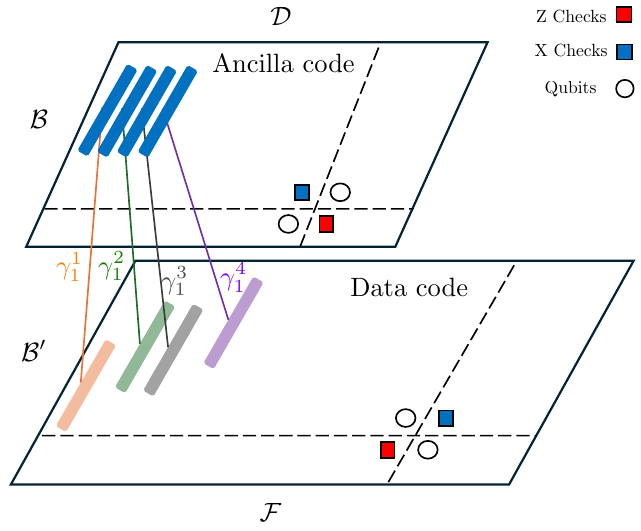}
     \caption{\textbf{Illustrative of high-rate surgery on a HGP data code using another HGP ancilla code}. The logicals to measure on the data code are supported on different columns and are indicated the bars in varying colors. The blue bars in the ancilla code represent a set of ancilla checks, whose stabilizer values are combined to infer the corresponding logical values in the data code. 
     According to Eq.~\eqref{eq:code_homomorphism}, we can distince, column-wise connectivites $\{\gamma_1^i\}$, resulting in a large pattern of logical measurements on the data code.
     % Under the assumption that $\mathcal{B}^\prime$ has translational symmetry, the connections, $\gamma_1^i$ for $i=1,2,3,4$, can be chosen as different shifted transversal connections for each column. 
     }
    \label{fig:hgp_illustrate}
\end{figure}

\noindent
\textbf{Algebraic construction for structured codes.}
A core challenge in designing high-rate surgery schemes is to design ancilla codes and chain maps that preserve the merged-code distance, since the soundness condition in Lemma~\ref{lem:merge_dist} can be hard to satisfy.
To begin, we provide an ancilla construction that provably preserves distance; we later analyze which logical measurements it can perform.
The ancilla is a hypergraph-product (HGP) code, denoted $\mc{A}=\mathrm{HGP}\sbkt{\mc{B},\mc{D}}$, formed as a tensor product of two base classical codes $\mc{B}\!:\,B_1 \xrightarrow{\partial_B} B_0$ and $\mc{D}\!:\,D_1 \xrightarrow{\partial_D} D_0$, where $\partial_B$ and $\partial_D$ are the classical check matrices.
Let $\mc{B}$ and $\mc{D}$ have parameters $[n_B,k_B,d_B]$ and $[n_D,k_D,d_D]$, respectively.
See Methods for a detailed definition of the HGP code.
Under this construction, the transpose of its $X$-check, $\partial_1^A$, is known to satisfy a soundness property~\cite{Campbell_2019} that is weaker than the required linear soundness in Lemma~\ref{lem:merge_dist}.
We close this gap by taking one base code, e.g., $\mc{D}$, to be an expander code.

\begin{theorem}[HGP ancilla with one expander base code is distance-preserving]\label{thm:HGP_exp_dist_preserve}
    Consider a surgery procedure equipped with a LDPC commuting diagram and a data code, $\mathcal{C}$, with distance $d$. Let the ancilla code $\mathcal{A} = \text{HGP}\sbkt{\mathcal{B}, \mathcal{D}}$. If $\min\sbkt{d_B, d_D}\geq d$ and $\mathcal{D}$ is an expanding code up to weight $t =\Theta\sbkt{d}$, the merge code distance is preserved as $d_{\text{merge}} = \Theta\sbkt{d}$.
\end{theorem}
With the fault tolerance guarantee, we now analyze HGP ancilla's IER. First, its kernel is
\begin{equation}
    \ker \partial^A_1 = \ker \partial_B \otimes \ker{\partial_D},
\end{equation}
which has $\dim \ker \partial^A_1 = k_B k_D$. In addition, the ancilla size $|\mc{A}| = \Theta\sbkt{n_B n_D}$. Assuming the existence of the commuting diagram, we arrive at
% For data code $\mc{C}$, if there exists a commuting diagram with chain maps $\Gamma_1$ and $\Gamma_0$ (specifying how $\mc{A}$ is attached to $\mc{C}$) such that $\Gamma_1$ does not significantly reduce the dimension of $\ker \partial^A_1$, we have
\begin{equation}
    r_M = \frac{\dim \Gamma_1 (\ker^A_1)}{|\mc{A}|} \sim \Theta\sbkt{\frac{k_B k_D}{n_B n_D}}.
\end{equation}
Thus, using high-rate classical codes $\mc{B}$ and $\mc{D}$ leads to high IER of the surgery scheme, which is formalized as a theorem below.
\begin{theorem}[High-rate surgery with HGP ancilla achieves constant overhead]
\label{theorem:algebraic_surgery}
    For qLDPC code $\mc{C}$, a surgery scheme with HGP ancilla $\mc{A} = \mr{HGP}(\mc{B}, \mc{D})$ achieves constant IER if (1) the base codes $\mc{B}$ and $\mc{D}$ have constant encoding rate and (2) there exist a commuting diagram with chain maps $\Gamma_1$ and $\Gamma_0$, where $\Gamma_1$ does not significantly reduce the kernel dimension $\partial^A_1$, i.e. $\dim \Gamma_1(\partial^A_1) = \Theta(\dim \partial^A_1)$. The scheme is distance-preserving if it additionally satisfies Theorem~\ref{thm:HGP_exp_dist_preserve}.
\end{theorem}

% Note that the first step in Theorem~\ref{theorem:algebraic_surgery} --- finding quantum code $\mc{C}$ that can be connected to $\mc{A}$ via chain maps $\Gamma_1$ and $\Gamma_0$ --- resembles the task of finding two quantum codes linked through a chain map for homomorphic measurements~\cite{Huang2023Homomorphic, Xu2025Homomorphic}, yet being a much simpler task, as we essentially only need to chain a chain map (homomorphism) for two classical codes associated with the $X$-checks of $\mc{A}$ and the $Z$ checks of $\mc{C}$, respectively:
% \begin{eqnarray}\label{eq:code_homomorphism}
%         \begin{tikzcd}[column sep=large, row sep=large]
%         {A_1} & {A_0} \\
%          {C_1} & {C_0}
%         \arrow["{\partial^{A}_1}", from=1-1, to=1-2]
%         \arrow["{\partial^{C}_1}", from=2-1, to=2-2]
%         \arrow["{\Gamma_1}", from=1-1, to=2-1]
%         \arrow["{\Gamma_0}", from=1-2, to=2-2]
% \end{tikzcd}.
% \label{eq:XZ_check_map}
% \end{eqnarray}
Our ancilla construction can be interpreted as a high-rate ``HGP extractor", as opposed to the low-rate graph-based extractors~\cite{he2025extractorsqldpcarchitecturesefficient}. Here, the philosophy of an extractor is that given a fixed ancilla with guaranteed fault-tolerance properties, one design its connectivity, $\Gamma_1$ and $\Gamma_0$, depending on the structures in target quantum codes and logical measurements.

To characterize the logicals that can be measured in parallel requires specification of the data code and a logical basis. While we defer a detail discussion of the general symmetries required to the appendix, the discussion is straightforward when we consider the HGP code as the data code, which admits a canonical logical basis. For example, consider data code $\mathcal{C} = \text{HGP}\sbkt{\mathcal{B}^\prime, \mathcal{F}}$. Due to the product structure of $\mc{C}$, any X-logical in the canonical basis is supported in one of the columns and taking the form of a codeword of $\mc{B}^\prime$. As a result, the challenge of finding a quantum code homomorphism in Theorem~\ref{theorem:algebraic_surgery} reduces to finding classical code homomorphisms between the classical base codes $\mc{B}$ and $\mc{B}^{\prime}$,
\begin{eqnarray}\label{eq:code_homomorphism}
        \begin{tikzcd}[column sep=large, row sep=large]
        {B_1} & {B_0} \\
         {B^\prime_1} & {B^\prime_0}
        \arrow["{\partial_{B}}", from=1-1, to=1-2]
        \arrow["{\partial_{B^\prime}}", from=2-1, to=2-2]
        \arrow["{\gamma^i_1}", from=1-1, to=2-1]
        \arrow["{\gamma^i_0}", from=1-2, to=2-2]
        % \arrow["{\gamma_0}", from=1-3, to=2-3]
\end{tikzcd},
\end{eqnarray}
where $\gamma_1^i$ connects the $i$-th column of $\mc{A}$ to a distinct column of $\mc{C}$, and we allow these column-wise connections $\{\gamma_1^i\}$ to be different. See Figure~\ref{fig:hgp_illustrate} for an illustration.
% In such cases, we arrange the connections to be column-to-column, see e.g. Figure~\ref{fig:hgp_illustrate}.

% As an example, as detailed in Methods, we show that we can apply the HGP ancilla to a HGP data code $\mathcal{C} = \text{HGP}\sbkt{\mathcal{B}^\prime, \mathcal{F}}$. 
% In such a senario, because of the product structure of $\mc{C}$, it becomes much easier to construct the chain maps to satisfy Eq.~\eqref{eq:XZ_check_map}: finding such a map on the quantum codes simply reduces to constructing chain maps on the base classical codes:
% \begin{eqnarray}\label{eq:code_homomorphism}
%         \begin{tikzcd}[column sep=large, row sep=large]
%         {B_1} & {B_0} \\
%          {B^\prime_1} & {B^\prime_0}
%         \arrow["{\partial_{B}}", from=1-1, to=1-2]
%         \arrow["{\partial_{B^\prime}}", from=2-1, to=2-2]
%         \arrow["{\gamma^i_1}", from=1-1, to=2-1]
%         \arrow["{\gamma^i_0}", from=1-2, to=2-2]
%         % \arrow["{\gamma_0}", from=1-3, to=2-3]
% \end{tikzcd},
% \end{eqnarray}
% where $\partial_{B^{\prime}}$ denotes the check matrix of $\mc{B}^{\prime}$.

Several known approaches to construct classical code homomorphisms include puncturing and augmentation~\cite{Xu2025Homomorphic} and code automorphisms~\cite{malcolm2025computingefficientlyqldpccodes,berthusen2025automorphismgadgetshomologicalproduct}. However, prior works utilize symmetries mostly through a direct lift from classical to quantum product codes. As a result, the available induced logical operations are limited by the strict product structure. For example, when $\mc{B}$ has an automorphism group $\text{Aut}\sbkt{\mathcal{B}}$, directly lifting it to quantum-code automorphism~\cite{malcolm2025computingefficientlyqldpccodes, berthusen2025automorphism} induces an automorphism group of order $\bigO{\abs{\text{Aut}\sbkt{\mathcal{B}}}k_D !}$. 
In surgery, we can instead have the column-wise addressing in Eq.~\eqref{eq:code_homomorphism} by e.g., using distince automorphisms of $\mc{B}$ as $\{\gamma_1^i$\}, resulting in an exponential pattern of logical measurements $\bigO{\abs{\text{Aut}\sbkt{\mathcal{B}}}^{k_D}}$ and high logical addressibility. 
% Such a symmetry can be combined with the code modification techniques to measure logicals in each column up to an automorphism, 
Such an addressibility is beyond the capabilities of the above-mentioned automorphism gates and the recently developed homomorphic measurements~\cite{Xu2025Homomorphic} for HGP codes.

We emphasize that algebraic scheme in Theorem~\ref{theorem:algebraic_surgery} is general and applicable to other codes. 
It is particularly straightforward to extend to codes with a generalized product structure, such as LP code~\cite{Panteleev_2022}, fiber bundle code~\cite{Hastings_2021}, and balanced product codes~\cite{Breuckmann_2021}. We leave the exploration of HGP ancilla tailored to these codes and their logical basis to future work.\\

\begin{figure*}
    \centering
    \includegraphics[width = 0.7 \textwidth]{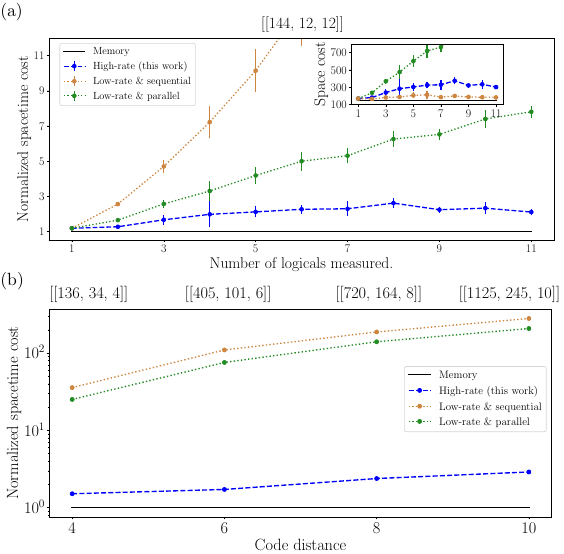}
     \caption{\textbf{Numerical analysis of the space-time costs of surgery schemes}. The spacetime costs are normalized by memory spacetime cost. The logical basis and logicals to measure are randomly sampled. The merge code degrees are limited to the data code degree plus two. Each data point is obtained through 6 sets of samples. (a) Cost for measuring varying number of logicals on Gross code. (Inset) the maximum space cost (with only data qubits). (b) Cost for measuring logicals in a set of high-rate spatially-coupled codes $\{[[n_i, k_i, d_i]]\}$ with increasing $k_i$ and $d_i$ (the code parameters are labeled on top). For each code instance, a random sample of $\floor{90 \%\times k_i}$ number of logicals are measured. }
    \label{fig:random_algo_results}
\end{figure*}

\noindent
\textbf{Randomized construction for generic codes.}
In the last section we presented an algebraic ancilla construction with sufficient soundness, enabling application to structured quantum codes with a provable distance guarantee.
In this section we go further: without any assumption on code structure, we can still perform high-rate surgeries in practice using a randomized construction for the ancilla system $\mc{A}$.
The heart of the approach is a randomized algorithm with optimized heuristics, motivated by prior randomized constructions of low-rate, graph-based ancillas~\cite{williamson2024lowoverheadfaulttolerantquantumcomputation,Cross2024ImprovedSurgery}.
At a high level, the construction proceeds as follows:
\begin{enumerate}
    \item \textbf{Initialization:} Construct ancilla $X$-checks $A_1$ that transversally connect to the union of all generating logicals. Add low-degree qubits $A_0$ and low-degree $Z$-checks $A_{-1}$ so as to measure the target logical space defined in Lemma~\ref{lem:log_action}.
    \item \textbf{Boost expansion:} At each step, sample qubits connected to the $X$-checks under the constraint that the induced logical measurements are unchanged. This typically boosts soundness.
    \item \textbf{Reduce check degrees:} Add either $Z$- or $X$-checks depending on the current sparsity of the parity checks.
    \item \textbf{Iterate:} Repeat steps 2 and 3 until the distance is sufficient.
\end{enumerate}

We remark that previous randomized graph algorithms do not extend straightforwardly to high-rate surgery.
Key reasons include the lack of hypergraph analogs of decongestion and cellulation~\cite{Freedman_2021}.
We design an algorithm that boosts the dressed distance with convergence guarantees: it adaptively alternates between boosting expansion and reducing average degree.
Another crucial ingredient is a local greedy selection rule to evaluate and choose randomly sampled qubits so as to maximize the expansion gain.
For a full discussion, see the Appendix.

In practice, finite increases in connectivity matter. Surgery-based protocols typically incur a minimum increase of two in the check degree. In our simulations, we constrain the algorithm to the same benchmark. As a result, we restrict to transversally connected chain maps $\Gamma_0$ and $\Gamma_1$.

We next apply this randomized construction to the $[[144,12,12]]$ Gross code~\cite{Bravyi2024Nature} and to new instances of spatially coupled (SC) codes~\cite{Yang2025SCQLDPC}.
Both codes possess strong symmetries that could be leveraged to build structured ancillas.
Nevertheless, to demonstrate generality, we do not exploit any symmetry: we measure a randomly sampled set of logical operators with a randomly generated logical basis.

As shown in Fig.~\ref{fig:random_algo_results}(a), using the Gross code as an example, low-rate surgery schemes—implemented sequentially~\cite{yoder2025tourgrossmodularquantum} with a graph-based ancilla of rate $r_M=\tilde{\mc{O}}(1/d)$ over $t$ logical cycles, or in parallel with $t$ disjoint ancillas~\cite{Cowtan2025ParallelSurgery}—exhibit rapidly increasing space-time overhead as the number of logical measurements $t$ grows.
The total overhead can exceed $\approx 10\times$ the memory cost when $t$ approaches the number of logical qubits $k=12$.
Moreover, in both sequential and parallel implementations, either the spatial or the temporal overhead becomes prohibitive.
In contrast, the high-rate scheme with a randomized hypergraph ancilla performs all $t$ measurements within a single logical cycle, using an ancilla no larger than the memory code (i.e., $|\mc{A}|\le|\mc{C}|$).
Consequently, the resulting space-time overhead remains within a factor of two of the memory cost, even for large $t$.

% In Figure.~\ref{fig:random_algo_results} we present the spacetime cost of high-rate surgery compared to the low-rate surgery with sequential and parallel designs for the Gross code, $[[144, 12, 12]]$. The memory baseline is also shown for comparison. One can see that the spacetime cost of our high-rate scheme converge with the low-rate schemes at small number of measurements as expected, but its cost smoothens to a constant scaling as we increase the number of logicals measured. To compare, the overhead of the low-rate schemes increases linearly in $t$. As a rule-of-thumb, the scheme incurs significant overhead whenever $t \gtrsim n/d$. In the example, its spacetime cost peaks at around $10^3$ when most logicals are measured. For the high-rate scheme, the trend is qualitatively different due to two competing effects: as the number of measured logicals increase, the connected region \qxc{connected region not defined} to the data code increases, leading to more possibility in reducing the distance of unmeasured logicals; at the same time, the number of unmeasured logicals decreases with the number of measured logicals. Thus, less logicals need to be protected as we approach the right end of the figure.

\begin{table}
    \centering
    \begin{tabular}{c|c|c}
        \hline \hline
        Parameters & Rate & Degrees\\
         \hline 
         $[[136, 34, 4]]$ & $25\%$ & $(5, 8)$\\
        %  \hline 
         $[[405, 101,6]]$ & $24\%$ & $(6, 9)$ \\
         $[[720, 164, 8]]$ & $22\%$ & $(4, 6)$\\
         $[[1125, 245, 10]]$ & $21\%$ & $(6, 9)$ \\
         \hline \hline 
    \end{tabular}
    \caption{\textbf{Summary of spatially-coupled code parameters}. The check degrees shown are the qubit degrees and check weights respectively. Distances are estimated through GAP~\cite{Pryadko2022} and should be taken as upper bounds.}
    \label{tab:SC_code_params}
\end{table}

% In the inset of Figure.~\ref{fig:random_algo_results}, we show the spatial cost alone for the schemes. Clearly, while the low-rate parallel scheme outperforms the low-rate sequential scheme in spacetime cost, it has a forbidding large spatial overhead since it requires all qubits to be simultaneously present. This is true for the parallel scheme whenever the sum of the logicals' weights are much larger than $n$.
% \begin{figure}[t!]
%     \centering
%     \includegraphics[width = 0.5 \textwidth]{SC.pdf}
%      \caption{Spacetime overhead for measuring randomly chosen sets of joint logicals. We randomly select $\floor{90 \%\times k}$ number of logicals to measure for each code. The codes are the ones in Tab~\ref{tab:SC_code_params} corresponding to each distance. \todo{Do we need to adjust the definition of spacetime cost in this figure? Did we divide by an additional $90\%$?} The code degree increase is maximum 2 and 6 samples are drawn for each dara point.}
%     \label{fig:SC_code_cost}
% \end{figure}

\begin{figure}[t!]
    \centering
    \includegraphics[width = 0.5 \textwidth]{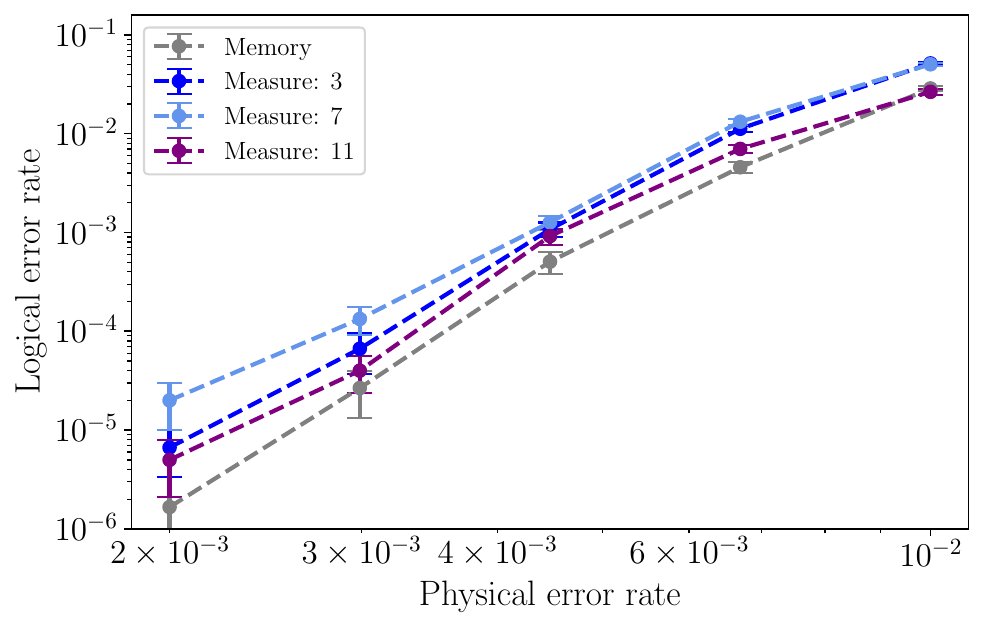}
     \caption{\textbf{Memory logical error rate of the $[[144, 12, 12]]$ Gross code and its merge codes during surgery}. Simulations are performeed under standard circuit-level noise, excluding idling errors. The merge codes shown are constructed through the randomized, high-rate surgery scheme for measuring 3, 7, and 11 of its logicals under a randomly chosen logical basis. The logical error rates reported are averaged per logical qubit over $3$ code cycles. BP-LSD decoder~\cite{hillmann2024localized} with the same decoder parameter is used to decode all cases. 
     }
    \label{fig:gross_LER}
\end{figure}

Next, to demonstrate constant overhead using high-rate surgeries with constant IER $r_M$, we consider a set of spatially coupled codes with similar encoding rates (between $20\%$ and $25\%$) and increasing distances (from $4$ to $10$).
See Tab.~\ref{tab:SC_code_params} for a summary of their code parameters and Appendix~\ref{sec:app_SC_codes} for their construction.
As shown in Fig.~\ref{fig:random_algo_results}(b), for a $[[n,k,d]]$ SC code the space-time cost relative to the memory overhead for measuring a randomly sampled set of $\floor{90\% \times k}$ logical operators increases very slowly and remains within a factor of $2$ as $d$ increases.
In contrast, low-rate schemes have overhead that grows quickly with $d$, consistent with their decreasing IER.

We highlight that for the largest SC code with $245$ logical qubits, our scheme performs $\approx 200$ randomly sampled logical measurements in parallel with roughly twice the memory overhead.
Here, the memory overhead is already extremely low due to the high encoding rate.
This demonstrates the potential of these codes, combined with our technique, for executing complex, addressable, large-scale logical computations with low overhead.
Although these SC codes have extremely high rate, their logical operations and code structures remain largely unexplored.
These results therefore highlight the power of our randomized construction: addressable, low-cost logical operations are generated automatically with no prior assumptions on the code.

We conjecture that the slight increase in overhead observed in Fig.~\ref{fig:random_algo_results}(b) reflects limitations of our current numerical algorithms, which could be further improved to approach genuinely constant scaling; we leave this optimization to future work.

Finally, we simulate the logical error rate for the Gross code to show that the high-rate surgery scheme does not significantly impact data-code performance.
In particular, we simulate the memory performance of the merged code and compare it with that of the data code, decoding both with the BP-LSD decoder~\cite{hillmann2024localized} for consistency.
Similar to prior work~\cite{yoder2025tourgrossmodularquantum}, we observe an increase in logical error rate, but the magnitude is acceptable.
Moreover, the merged code often has significantly increased $Z$-distance, leading to orders-of-magnitude bias between $Z$ and $X$ errors.
This presents opportunities for further optimization, such as a randomized distance-balancing algorithm.

\section{Discussion and outlook}

In this paper, we provide a new framework and concrete constructions for high-rate code surgeries.
By formalizing the logical action in a homological framework and giving a sufficient condition for distance preservation, we establish a general recipe for constructing ancilla patches that exploit the dense logical encoding of quantum LDPC codes.
The framework motivates explicit ancilla constructions that achieve provably constant overhead when suitable structure is present in the data code.
We further present randomized algorithms for finite-size codes that are agnostic to code structure.
Remarkably, the computational cost is typically within a factor of two of the memory cost.
Taken together, these results demonstrate the potential of high-rate surgery by showing its low space-time cost both asymptotically and for practical code instances.

Several promising research directions follow.
As noted, surgery schemes are essentially a weight-reduction method in which one or more logical operators are the targets.
Analogous to low-rate surgery, prior approaches to weight-reducing parity checks~\cite{hastings2016weightreductionquantumcodes, hastings2023quantumweightreduction, Sabo_2024} apply the same procedure iteratively to all checks above a weight limit.
A natural extension of our work is a parallel weight reduction of stabilizers, which we expect to yield resource savings similar to those observed here.

Another direction is to extend the framework to single-shot quantum error correction and related work~\cite{Campbell_2019, Quintavalle_2023, zhou2025lowoverheadtransversalfaulttolerance, PhysRevLett.133.240602, baspin2025fastsurgeryquantumldpc, xu2025batchedhighratelogicaloperations}, potentially further reducing temporal overhead.
Exploring optimized ancilla designs based on other good code families—including bivariate-bicycle codes and quantum expander codes—may lower constant factors and relax the need for product-structure symmetries.
Finally, our randomized algorithm is a proof-of-principle demonstration; improved heuristics or combinatorial optimization should further reduce resource costs, especially for large codes.

\section*{Acknowledgments}
We thank Gideon Lee, Han Zheng, Pei Zeng for their helpful discussions. G.Z. and L.J. acknowledge support from the ARO(W911NF-23-1-0077), ARO MURI (W911NF-21-1-0325), AFOSR MURI (FA9550-21-1-0209, FA9550-23-1-0338), DARPA (HR0011-24-9-0359, HR0011-24-9-0361), NSF (ERC-1941583, OMA-2137642, OSI-2326767, CCF-2312755, OSI-2426975), Packard Foundation (2020-71479), and the Marshall and Arlene Bennett Family Research Program. This material is based upon work supported by the U.S. Department of Energy, Office of Science, National Quantum Information Science Research Centers and Advanced Scientific Computing Research (ASCR) program under contract number DE-AC02-06CH11357 as part of the InterQnet quantum networking project. Q.X. is funded in part by the Walter Burke Institute for Theoretical Physics at Caltech.

% Overall, this work bridges a major gap between low-overhead quantum memory and low-overhead computation, moving the field closer to the goal of constant-overhead, fault-tolerant quantum computing.

\clearpage
\newpage
\section{Methods}

\textbf{Canonical basis.} Here we give an overview of the canonical basis for completeness. We define a basis for a classical code codeword, $[n, k]$, if it possess the following form
\begin{eqnarray}
        G &=& \begin{pmatrix}
            I_k & Q
        \end{pmatrix}\\
        H &=& \begin{pmatrix}
            Q^T & I_{n-k}
        \end{pmatrix}
\end{eqnarray}
up to column permutations. Here $I_k$ is a $k\times k$ dimensional matrix. We say the bit support of the $I_k$ submatrix are the information bits, whose union is the information set. Worth noting that the information bits are not unique. Similarly we can introduce a canonical basis form~\cite{Gottesman1997Stabilizer} for quantum codes 
\begin{eqnarray}
        G_X &=& \begin{pmatrix}
            I_k & Q_X & 0
        \end{pmatrix}\\
        G_Z &=& \begin{pmatrix}
            I_k & 0 & Q_Z
        \end{pmatrix}
\end{eqnarray}
and we can similarly define the information (qu)bits. In general, the canonical basis is not a minimum weight basis.

\textbf{Classical code homomorphisms} Consider a classical homomorphism as the following commuting diagram,
\begin{eqnarray}\label{eq:log_check_mod}
        \begin{tikzcd}[column sep=large, row sep=large]
        {B_1} & {B_0} \\
         {B^\prime_1} & {B^\prime_0}
        \arrow["\partial^B", from=1-1, to=1-2]
        \arrow["\partial^{B^\prime}", from=2-1, to=2-2]
        \arrow["{\gamma_1}", from=1-1, to=2-1]
        \arrow["{\gamma_0}", from=1-2, to=2-2]
\end{tikzcd},
\end{eqnarray}
The commuting diagram implies $\gamma_1  \partial^{B^\prime} =\partial^B\gamma_1$, which induces a logical action, $w$, given by
\begin{eqnarray}
        G_{B} \gamma_1^T = w G_{B^\prime}
\end{eqnarray}
$w$ depends on the logical basis chosen. Here we given two specific examples.
\begin{enumerate}
    \item Code modification: the modification technique here is referred to as logical checks in our work. It is similar to the puncture and augmentation techniques in classical coding theory. Consider the choice of $\mathcal{B}$ as
    \begin{eqnarray}
        \partial^B = \begin{pmatrix} \partial^{B^\prime}\\ H_L \end{pmatrix}
\end{eqnarray}
and chain maps 
\begin{eqnarray}
        \gamma_1 = I_1, \gamma_0 = \begin{pmatrix} I_0 & 0\end{pmatrix}
\end{eqnarray}
where $I_{0,1}$ and $0$ should be interpreted as identity and zero block matrices with matching dimensions. The kernel space follows as $\ker{\partial^{B}} = \ker{\partial^{B^\prime}} \cap \ker{H_L}$. Suppose $\mathcal{B}$ and $\mathcal{B}^\prime$ are the ancilla and data codes respectively, $H_L$ reduce the ancilla code kernel space to a subspace of the data code. Thus, it is effectively a logical check that filters unwanted logicals. For clarity, consider an exaple where weadopt the canonical basis as $G_{B^\prime} = \begin{pmatrix} I_{k} & Q \end{pmatrix}$ and assume the logical checks only have support on the information set as $H_L = \begin{pmatrix}
    H_I & 0
\end{pmatrix}$. The kernel space follows as 
\begin{eqnarray}
        G_{B} = \ker{H_I}\begin{pmatrix} I_{k} & Q \end{pmatrix}.
\end{eqnarray}
Thus, we tailor the kernel space by designing the kernel space of $H_I$. In the context of surgery, $H_I$ correspond to the unencoded circuit's sparsity. When we assume it is a shallow circuit, $H_I$ is necessarily LDPC. As a simple example, consider 
\begin{eqnarray}
    G_{B^\prime} = \begin{pmatrix} 1 & 0 & 0 & a_1^T\\
    0 & 1 & 0 & a_2^T\\
    0 & 0 & 1 & a_3^T\end{pmatrix}, H_I = \begin{pmatrix}
        1 & 0 & 0\\
        0 & 1 & 1
    \end{pmatrix}
\end{eqnarray}
where one can understand the first and second lines of $H_I$ realizes puncturing the first logical and coupling the second and third logicals respectively. As a result, we realize 
\begin{eqnarray}
    G_{B} = \begin{pmatrix} 0 & 1 & 1 & a_2^T + a_3^T
    \end{pmatrix}
\end{eqnarray}
\item Code automorphisms: automorphism refer to the special case where $\mathcal{B}^\prime = \mathcal{B}$ and is often a result of the code's symmetry. The automorphism group, $\text{Aut}\sbkt{\mathcal{B}^\prime}$, consists of $\gamma_1$ such that there exists $\gamma_0$ where the diagram commutes. The automorphism group can be generated by permutations and the induced logical actions are affine, $w\in \text{GL}_{k_B}\sbkt{\mathbb{F}_2}$. The restriction to affine actions makes it not universal for classical codes.
\end{enumerate}
We highlight that the commuting diagram can be superposed when $\mathcal{B}, \mathcal{B}^\prime$ are the same through $\gamma_1 = \gamma_1^a + \gamma_1^b, \gamma_0 = \gamma_0^a + \gamma_0^b$, which is different from group compositions. This can, e.g., generate logical actions beyond affine ones with superposed automorphism gadgets. In addition, any code homomophism $\mathcal{B}\to\mathcal{B}^\prime$ can be composed with any automorphisms $\mathcal{B}^\prime\to\mathcal{B}^\prime$. For example, one can combine logical checks with translational automorphisms, as we have demonstrated in the main text.

\textbf{Hypergraph product code (HGP)} For completeness, we provide the definition of $\text{HGP}\sbkt{\mathcal{B}, \mathcal{D}}$. As length-2 chain complex, they can be written as
\begin{eqnarray}
        \begin{tikzcd}[column sep=large, row sep=large]
        {} & {B_1\otimes D_0} & {} \\
        {B_1\otimes D_1} &  & {B_0\otimes D_0}  \\
        {} & {B_0\otimes D_1} & 
        \arrow["I_{B_1}\otimes \partial^D", from=2-1, to=1-2]
        \arrow["\partial^B \otimes I_{D_1}", from=2-1, to=3-2]
        \arrow["I_{B_0}\otimes \partial^D", from=3-2, to=2-3]
        \arrow["\partial^B \otimes I_{D_0}", from=1-2, to=2-3]
\end{tikzcd}
\end{eqnarray}
Through Kunneth theorem, it has parameters
\begin{eqnarray}
                n &=& n_{B_{1}}n_{D_0}+ n_{B_{0}}n_{D_1}\\
                k &=& k_B^T k_D + k_B k_D^T\\
                d_X &=& \min\sbkt{d_B, d_D}\\
                d_Z &=& \min\sbkt{d_B^T, d_D^T}.
\end{eqnarray}
Note that our definition is different from some definitions in literature by a code transpose.

\textbf{Surgery procedure} Consider a general surgery scheme with the following diagram 
\begin{eqnarray}
        \begin{tikzcd}[column sep=large, row sep=large]
        {A_1} & {A_0} & {A_{-1}} \\
        {C_2} & {C_1} & {C_0}
        \arrow["{\partial^A_1}", from=1-1, to=1-2]
        \arrow["{\partial^A_0}", from=1-2, to=1-3]
        \arrow["{\partial^C_2}", from=2-1, to=2-2]
        \arrow["{\partial^C_1}", from=2-2, to=2-3]
        \arrow["{\Gamma_1}", from=1-1, to=2-2]
        \arrow["{\Gamma_0}", from=1-2, to=2-3]
        % \arrow["{\gamma_0}", from=1-3, to=2-3]
\end{tikzcd}
\end{eqnarray}
One begins with a data code $\mathcal{C}$, and we follow the convention of X-Q-Z vector spaces in the main text. The conventional surgery procedure consists of the follow steps:
\begin{enumerate}
    \item Initialize qubits in $A_0$ in $\ket{0}^{\otimes \abs{A_0}}$.
    \item Measure X-checks in $A_1, C_2$ and Z-checks in $C_0, A_{-1}$. Repeat for $d$ rounds.
    \item Transversally measure qubits in $A_0$ in Z-basis. Apply Pauli frame correction to track the Z-stabilizer values in $C_0$.
\end{enumerate}
The logical information is measured and can be inferred through step 2. Since the qubits are initialized in the Z-basis, X-stabilizers in $A_1$ are not random if and only if they do not have support on $A_0$. For the details on the Pauli correction see Ref.~\cite{williamson2024lowoverheadfaulttolerantquantumcomputation}.

\textbf{Generalizations from graph to hypergraph subroutines} 
When the ancilla system is formed through disjoint patches of $\dim{\ker{\partial_1^A}} = 1$, each patch can be understood as a graph. Naturally, the generalization to high-rate ancilla translates to generalization from graphs to hypergraphs. Such a translation has two main difficulties: $Z$ logicals of the ancilla system, $G_Z^A$, might not admit a low-weight generating basis, and the construction of LDPC commuting diagram such that the correct logicals are measured, $\Gamma_1\sbkt{\ker{\partial_1^{A}}}= \spanv{L_M}$, where $L_M$ is a generating basis for the logicals to measure. In the following we address each of these problems. 

In hypergraphs, many of the graph subroutines are no longer valid. The most notable ones being cellulation and decongestion, both of which are key components in weight-reducing $G_Z^A$, which is the graph's cycle basis. One observation is that when data code has Z-metachecks, they often serve as a reverse-directioned surgery that effectively measures out (part of) $G_Z^A$. While this does not provide useful information, it reduces the number of gauge logicals in the merge code. A similar observation was initially made in Ref.~\cite{Cross2024ImprovedSurgery} for HGP codes. In rare occasions where some large $G_Z^A$ remain unmeasured, we accept them as gauge logicals instead of attempting to cellulate and decongest. Our choice is motivated by heuristics: while cellulation and decongestion algorithms can be naively generalized to hypergraphs, they no longer have a guarantee on the resulting code's degree. In fact, we find it challenging to suppress both X and Z check degrees for decently large codes, especially when we want to limit the increase in code degree to 2. In cases where the merge code is a subsystem code, we observe a modest increase in the overhead. 

Another challenge is posed by measuring the right logicals, i.e. the measured operator space $\spanv{\Gamma_1\sbkt{\ker{\partial_1^A}}}$ can contain unwanted logicals. For low-rate surgery this is initially dealt with by assuming the logical to be irreducible~\cite{Cohen2022LDPCSurgery,Cross2024ImprovedSurgery}. In later works, the ancilla system can be constructed by randomly adding edges to the graph, which effectively eliminates all other kernel dimensions except the target logical. However, these techniques and assumptions are only valid for low-rate surgery. There is no universal gaurantee that for arbitrary logical basis, one can measure exactly an arbitrary logical set through a LDPC choice of $\Gamma_1$ and $\ker{\partial_1^A}$. 

For simplicitly, we consider $\Gamma_{0,1}$ to be of transversal-type, but our conclusion is valid for any connections that are LDPC. Following the discussion on logical checks, a reasonable assumption is that the logicals to measure admits a LDPC logical check for the unencoded circuit to be of depth-$\bigO{1}$. Depending on the choice of logical basis and construction of commuting diagram, we discuss case-by-case as following:
\begin{itemize}
    \item The basis is canonical, and the classical code homomorphism is constructed by puncturing. Here puncturing refers to taking the largest union necessary. In general codes this means $B_1^\prime = B_1$. In specific codes such as HGP code $B_1^\prime$ can only be the size of a column because of the product structure. In this case we are gauranteed to have both a LDPC homomorphism that is measures correctly and a stabilizer code with no gauge logicals.
    \item The basis is canonical, but the classical code homomorphism is not punctured. For example, the homomorphism can be obtained through taking the union. In this case we still have LDPC gaurantee on measuring the correct logicals, but no longer the stabilizer code gaurantee.
    \item The basis is not canonical. We don't have a gaurantee in this scenario. However, in practice we see that for decent-sized codes and sampled joint logicals to measure, one finds logical checks that are low in degree in most cases. This is only violated in rare instances. We want to point out that in simulations we never encounter large-degree logical checks if a random basis is chosen. This can be understood intuitively since these basis have relatively large weights compared to $d$, which means that they have more qubits per logical that one can place logical checks on such that the overall weight is reduced.
\end{itemize}

\textbf{Randomized algorithm for ancilla construction} In the last section we discussed from a high-level why a randomized algorithm can still retain degree for high-rate surgery. Here, we give a more detailed procedure description. For the full protocol, see appendix. For a general quantum code and a randomized set of logical operators $L_M$, we perform the following steps
\begin{enumerate}
    \item Take the union of the support of $L_M$, i.e. choose a set of qubits $Q\subseteq C_1$ such that for every logical operator $l_i\in L_M$ we have $l_i\subseteq Q$. Moreover, for every qubit in $Q$, it is in the support of at least one logical operator in $L_M$. 
    \item Construct Z-checks in $B_1$ that are transversally connected to $Q$. Construct qubits in $B_0$ that are transversally connected to $\text{Supp}\lbkt{\sbkt{H_Z\vert_Q}^T}$, i.e. we first restrict the Z-checks to $Q$ then delete all checks that have empty rows.
    \item Optimize for low-weighted logical checks, which corresponds to adding qubits in $B_0$ that are connected to $B_1$. The optimization is done through a distance estimation algorithm in GAP~\cite{Pryadko2022}. The check after this step is denoted as $\partial_B^{\prime \prime}$
    \item Verify if the Z-logicals of $\mathcal{A}$ are in the stabilizer group. For those that are not, implement if one can find a low-weight basis.
    \item Sample low-weight (dressed) logicals of the merged code that are below the target distance, $d$. Sample valid ancilla qubits that do not modify $\ker{\partial_1^A}$. A qubit is valid if its connections are low-weight. Assign a score to each qubit based on how many low-weight X-logicals they anticommute with. Implement the one with the largest score.
    \item If there is no valid qubits with non-zero score, extend the ancilla by one layer by adding a full copy of the checks in $B_1$ and connect them to the previous layer through transversal qubits.
    \item Repeat the above procedure until target distance achieved.
\end{enumerate}
We present a more detailed analysis of the algorithm and its convergence in the Appendix.

\textbf{Measurement of mixed-type Paulis} For simplicity, we have restricted to measurement of logicals of single Pauli type in a CSS code. Now, we comment on how it generalizes to general logical measurements. Suppose we consider CSS codes and measurement of mixed types of logicals. One can always write the logical measurement set as a commuting set of logical Paulis. However, it is not generally guaranteed that all composite X-logicals and Z-logical commute with each other. This poses a challenge if we construct the X and Z ancilla system separately since we would be measuring operators that anticommute with the stabilizer group, which goes beyond our framework. A remedy is to employ the canonical logical basis~\cite{gottesman1997stabilizercodesquantumerror, Cowtan2025ParallelSurgery} and generalize the framework in Lemma~\ref{lem:log_action} to mixed-type Pauli connectivities. Another remedy is to consider twist-free surgery and Y resource states, similar to Refs.~\cite{Cowtan2025ParallelSurgery,Guo2025TimeEfficient}. See Appendix for an extended discussion.

% \bibliographystyle{apsrev4-2}
% \bibliography{ref}


\clearpage

\pagebreak
% \break
\appendix 

\widetext
\begin{center}
\textbf{\large Supplemental Material}
\end{center}

\section{Notation and background \label{app:notation}}
% \subsection{Notation}
We start by introducing some notations. Throughout this work, $[n]$ is used to denote the set of integers $\set{1,\dots, n}$. Without specification, we assume vectors are column vectors. For any vector $v\in \mathbb{F}_2^n$, its support is defined as $\text{supp}\sbkt{v} = \set{i \vert v_i= 1}$, where $v_i$ denotes its entry at index $i$. The supports of a matrix, $A$, are defined correspondly as 
\begin{eqnarray}
                \text{rowsupp}\sbkt{A} &=& \bigcup_i \text{supp}\sbkt{A_{:, i}}\\
                \text{colsupp}\sbkt{A} &=& \bigcup_i \text{supp}\sbkt{A_{i, :}}
\end{eqnarray}
i.e. column support represent all the columns that any rows of $A$ has nontrivial support. We follow the convention that $\spanv{A}$ denotes the vector space spanned by row vectors of $A$. With a slight abuse of notation, we write for a given vector space $V$ and linear map $\gamma$ that $\gamma \sbkt{V} := \set{ \gamma v\vert v\in V}$.

A chain complex, $\mathcal{C}$, is a sequence of vector spaces $\set{C_i}$ with linear maps known as boundary operators, $\partial_{i}: C_i\to C_{i-1}$. A chain complex composed of $n+1$ vector spaces is called a length-$n$ chain complex. The boundary maps satisfy the boundary condition $\forall i\in [n]$,
\begin{eqnarray}
                \partial_{i-1} \partial_{i} = 0.
\end{eqnarray}
The $i$th homology group is defined as $H_i\sbkt{\mathcal{C}}:= \ker{\partial_{i}}/ \im{\partial_{i+1}}$. One can define the cochain complex and cohomology groups analogously. A chain complex is exact at $C_i$ if $\dim{H_i\sbkt{\mathcal{C}}} = 0$ or equivalently $\im{\partial_{i+1}} = \ker{\partial_i}$. 

Since we mostly work in $\mathbb{F}_2$, the notation $\abs{\cdot}$ is used to denote the Hamming weight of a vector. When acted on a matrix, it denote the maximum Hamming weight of all row and column vectors, i.e. the maximum matrix degree. 

A classical code is defined through a parity check matrix, $H$, and can be written as a length-1 chain complex as 
\begin{eqnarray}
        \begin{tikzcd}[column sep=large, row sep=large]
       {C_1} & {C_0}
        \arrow["{\partial}", from=1-1, to=1-2]
\end{tikzcd}
\end{eqnarray}
with $\partial = H$. The logical space is $\ker{H}$.

A quantum CSS stabilizer code~\cite{PhysRevA.54.1098, PhysRevLett.77.793} is a quantum code where the parity check matrices admit a generating basis that can be divided into X- and Z-type checks, $H_{X,Z}$. Since they follow the commutation relation, it can be described through a length-2 chain complex as 
\begin{eqnarray}
        &\begin{tikzcd}[column sep=large, row sep=large]
        {C_2} & {C_1} & {C_0}
        \arrow["{\partial_2}", from=1-1, to=1-2]
        \arrow["{\partial_1}", from=1-2, to=1-3]
\end{tikzcd}&
\end{eqnarray}
Since a chain complex and a CSS code are in one-to-one correspondance, we refer to them interchangeably depending on the context. In this work we always fix
\begin{eqnarray}
        \partial_2 = H_X^T, \partial_1 = H_Z
\end{eqnarray}
so that the vector spaces $C_l$ for $l = 0, 1, 2,$ denotes the $Z$ checks, the qubits, and the $X$ checks, respectively. 
% $\mbkt{\text{X check}\to \text{Qubit} \to \text{Z syndrome}}$. 
The stabilizer groups of a code is defined as $S_X := \spanv{H_X}$ and $S_Z := \spanv{H_Z}$ . The X-logicals are in $L_X = H_1\sbkt{\mathcal{C}}$. The logical dimension of the code is $k:= \dim{H_1\sbkt{\mathcal{C}}}$, with the X-distance being
\begin{eqnarray}
                d_X = \min_{v\in H_1\sbkt{\mathcal{C}}} \abs{v}
\end{eqnarray}
and $d_Z$ is defined analogously for logical $Z$ operators in the first co-homology group. These definitions correspond to the (co)systoic distances of a chain complex. The code distance is $d = \min\sbkt{d_X, d_Z}$. When $k=0$, we take $d\to\infty$ by convention.

A quantum CSS subsystem code~\cite{PhysRevLett.95.230504} is defined through a non-abelian gauge group, and the stabilizer group is the center of the gauge group. We can write a generating set for the stabilizer group similarly as $H_X$ and $H_Z$, and the gauge logicals are generated by $G_X$ and $G_Z$. Together, they generate the full gauge group. Without loss of generality, we let
\begin{eqnarray}
                G_{X, Z} H_{Z,X}^T &=& \mathbf{0}\\
                G_{X} G_{Z}^T &=& I_{g}
\end{eqnarray}
where $I_{g}$ has dimension equal to the number of gauge logicals. The dressed logicals, $L_{X,Z}$ are defined as the centralizer of $S_{X,Z}$ that are not in the gauge group. The dressed distance is the minimum weight of all dressed logicals.

% \section{CSS codes, chain complex, and commuting diagrams}

One recurring component in our work is commuting diagrams, which describes relations between chain complexes. Consider the following diagram
\begin{eqnarray}
        \begin{tikzcd}[column sep=large, row sep=large]
        &{C_i} & {C_{i-1}} &\\
         & {C^\prime_i} & {C^\prime_{i-1}}&
        \arrow["{\partial_{i}}", from=1-2, to=1-3]
        \arrow["{\partial^\prime_i}", from=2-2, to=2-3]
        \arrow["", from=1-1, to=1-2]
        \arrow["", from=1-3, to=1-4]
        \arrow["", from=2-1, to=2-2]
        \arrow["", from=2-3, to=2-4]
        \arrow["{\gamma_i}", from=1-2, to=2-2]
        \arrow["{\gamma_{i-1}}", from=1-3, to=2-3]
\end{tikzcd}
\end{eqnarray}
and we say it is commuting if and only if 
\begin{eqnarray}\label{eq:comm_def}
                \gamma_{i-1} \partial_i = \partial_i^\prime\gamma_i.
\end{eqnarray}
While being a simple result, we can conclude the condition for the existence of a commuting diagram in the following lemma.
\begin{lemma}[Condition for commuting diagrams]\label{lem:existence_com_diag}
        Given boundary operators $\partial_i$ and $\partial_i^\prime$, a commuting diagram exists if and only if either of the following holds,
        \begin{eqnarray}
                        \gamma_i\sbkt{\ker{\partial_i}}&\subseteq& \ker{\partial_i^\prime}\\
                        \gamma_{i-1}\sbkt{\im{\partial_i}}&\subseteq& \im{\partial_i^\prime}.
        \end{eqnarray}
\end{lemma}
\begin{proof}
        We focus on the first condition, and the proof for the second proceeds analagously. We prove $\gamma_i\sbkt{\ker{\partial_i}}\subseteq \ker{\partial_i^\prime}$ is a necessary and sufficient condition for the existence of a $\gamma_{i-1}$ that satisfies Eq.~\eqref{eq:comm_def}. The necessity is straightforward as for all $x\in \ker{\partial_i}$ we have 
        \begin{eqnarray}
                        \partial_i^\prime \gamma_i x = \gamma_{i-1} \partial_i x = 0,
        \end{eqnarray}
        so $\gamma_i x \in \ker{\partial_i^\prime}$. The converse is shown constructively. Consider $v_{i-1}\in \im{\partial_i}$ such that $v_{i-1} = \partial_i v_{i}$. The map is defined as $\gamma_{i-1}\sbkt{v_{i-1}}:= \partial_i^\prime \gamma_1 v_i$. For $u_{i-1}\in C_{i-1}\setminus \im{\partial_i}$, $\gamma_{i-1}$ can be assigned arbitrarily. Such a $\gamma_{i-1}$ is unique from the assumption, and it satisfies the commutation relations.
\end{proof}
With a commuting diagram, the maps $\gamma_i$ form a chain map. 

\section{Expansion properties}

In this work, the expansion properties of the parity check matrices are a key to gaurantee its distance properties. There are many notions of expansions, and the one most relevant to our purpose is soundness defined as below. We start with a slightly generalized notion of reduced distance that will be useful later.
\begin{definition}[Reduced distance]
     The reduced distance of a vector $x \in U$ with respect to a subspace $V\subseteq U$ and a linear map $\gamma: U \rightarrow U^{\prime}$ is defined as
    \begin{eqnarray}
            d_\gamma(x, V) := \min_{v\in V} \abs{\gamma \sbkt{x + v}}
    \end{eqnarray}
    where $\abs{\cdot}$ denotes Hamming weight, and $\gamma$ should be taken as the identity map when omitted.
\end{definition}
A simple inequality follows as
\begin{eqnarray}\label{eq:dist_ineq}
        d_\gamma(x, V) \leq \abs{\gamma} \cdot d(x, V)
\end{eqnarray}
where $\abs{\gamma}$ represents the weight of the matrix representation of $\gamma$.

\begin{definition}[Soundness; \cite{Campbell_2019}] \label{def:campbell_soundness}
        A map, $H: V\to V^\prime$, is $(t,f)$-sound for some integer $t$ and a function $f: \mathbb{Z} \rightarrow \mathbb{R}$ if for all $x\in V$ and $\abs{Hx}<t$, we have 
        \begin{eqnarray}
                        d\sbkt{x, \ker{H}} \leq f\sbkt{\abs{Hx}}
        \end{eqnarray}
\end{definition}

In our interest, $H$ is often the parity check matrix of a classical or quantum code. In this context, soundness have been shown to be sufficient conditions 
% for macroscopic energy barrier and success of 
for single-shot error correction if $t$ grows with $n$ and $f$ is a monotonically increasing function independent of $n$~
\cite{Campbell_2019}. 
Here $n$ denotes the system size.

The definition of soundness is later relaxed to confinement~\cite{PRXQuantum.2.020340}, where the weight constraint is put on the error weight $\abs{x}<t$ instead of syndrome weight. Confinement successfully explained single-shot properties of codes like quantum expander code~\cite{Leverrier_2015}, which are confined but not sound. Soundness is also closely related to (co-)boundary expansion notions of a chain complex. The major difference is in whether a limit is put on $\abs{Hx}$ and if one restricts to only considering linear functions for $f$. While there is a deep literature on expansions properties, we focus on the relevant definitions only.

In general, for surgery schemes to fully preserve distance, it requires the strong condition of soundness, which is when $f$ is a linear function. Moreover, while the above definition of soundness targets at a code family, we extend it to the linear soundness for a single code or linear map.
\begin{definition}[Linear soundness]
    A map $H: V\to V^\prime$ has soundness $\rho$ up to $t$ if for all $x $ with $\abs{H x}\leq t$, we have
    \begin{eqnarray}
            \rho \leq \frac{\abs{H x}}{d\sbkt{x, \ker{H}}}.
    \end{eqnarray}
\end{definition}

\section{General framework for code surgeries \label{app:general_framework}}

In this section and the following few, we derive main theorems and lemmas for surgeries focusing on CSS data code and measurements of $X$-type logicals. While such scenarios are already of interest for many practical codes and circuits, we remark on how our results generalize to more general measurements in Appendix~\ref{app:mixed_meas}. 

From here on, we fix the key notation of data (memory) code and ancilla code as $\mathcal{C}$ and $\mathcal{A}$ for clarity. Other curly math symbols are used for classical or quantum codes depending on the context. We assume the data code has code parameters $[[n_C, k_C, d_C = d]]$.

\subsection{General surgery procedure}

A surgery procedure typically describes the process of attaching additional qubits and checks to an initial data code $\mathcal{C}: \begin{tikzcd}[column sep=large, row sep=large]
        {C_2} & {C_1} & {C_0}
        \arrow["{\partial^C_2}", from=1-1, to=1-2]
        \arrow["{\partial^C_1}", from=1-2, to=1-3]
\end{tikzcd}$. This can be used for logical measurements or simply modifying/expanding codes [CITE]. In our work, we focus on the case of measurements, naming the initial system as data system and the additioanal as the ancilla system. Since the goal is only to perform measurements, the outcome of the surgery procedure consists of classical bits of logical information and the data code. Therefore, the procedure should be code-preserving and consists of the following three steps:
\begin{enumerate}
    \item \textbf{Ancilla code initialization: } 
    an ancilla system is described by the following chain complex:
\begin{equation}
    \begin{tikzcd}[column sep=large, row sep=large]
            {A_1} & A_0 & {A_{-1}} 
        \arrow["{\partial^A_1}", from=1-1, to=1-2]
        \arrow["{\partial^A_0}", from=1-2, to=1-3]
        % \arrow["{\gamma_0}", from=1-3, to=2-3]
\end{tikzcd},
\end{equation}
where $A_{l}$ for $l = -1, 0, 1, $ denotes the $Z$ checks, the qubits, and the $X$ checks of the ancilla system, respectively.
We initialize the ancilla qubits $A_0$ in $\ket{0}$ states.
%     \begin{eqnarray}\label{eq:ancilla_initialize}
%         \begin{tikzcd}[column sep=large, row sep=large]
%         {} & {\ket{0}^{\otimes A_0}} & {}\\
%             {C_2} & {C_1} & {C_{0}} 
%         \arrow["{\partial^C_2}", from=2-1, to=2-2]
%         \arrow["{\partial^C_1}", from=2-2, to=2-3]
%         % \arrow["{\gamma_0}", from=1-3, to=2-3]
% \end{tikzcd}
% \end{eqnarray}
    \item \textbf{Merge: } 
    % additional X and Z checks are performed following the ancilla code checks matrices and chain maps, $\Gamma_{0,1}$. The original data Z-checks in $C_0$ are modified accordingly.
    measure joint $Z$ checks $A_{-1} \cup C_0$ and $X$ checks $A_1 \cup C_2$ on the ancilla and data qubits $A_0 \cup C_1$ according to the following commuting diagram: 
    \begin{eqnarray} \label{eq:commute_diagram_app}
        \begin{tikzcd}[column sep=large, row sep=large]
        {A_1} & {A_0} & {A_{-1}} \\
        {C_2} & {C_1} & {C_0}
        \arrow["{\partial^A_1}", from=1-1, to=1-2]
        \arrow["{\partial^A_0}", from=1-2, to=1-3]
        \arrow["{\partial^C_2}", from=2-1, to=2-2]
        \arrow["{\partial^C_1}", from=2-2, to=2-3]
        \arrow["{\Gamma_1}", from=1-1, to=2-2]
        \arrow["{\Gamma_0}", from=1-2, to=2-3]
        % \arrow["{\gamma_0}", from=1-3, to=2-3]
\end{tikzcd},
\label{eq:merged_code_complex}
\end{eqnarray}
where the ancilla $X$ checks $A_1$ are also connected to the data qubits via $\Gamma_1$ and the original $Z$ checks $C_0$ on the data code are extended to the ancilla qubits via $\Gamma_0$. The $X$ and $Z$ checks commute if and only if the above diagram commutes. We refer to the CSS code defined in Eq.~\eqref{eq:merged_code_complex} (see Ref.~\cite{Ide_2025} for more details) as the merged code. 
\item \textbf{Split:} transversally measure the ancilla qubits in $A_0$ in $Z$-basis, followed by a Pauli $X$ correction on the data code to fix the $Z$ checks $C_0$ to $+1$~\cite{williamson2024lowoverheadfaulttolerantquantumcomputation}.
\end{enumerate}
At the heart of the construction is the commmuting diagram described in Eq.~\eqref{eq:commute_diagram_app}, or equivalently
\begin{eqnarray}
        \Gamma_0 \partial_1^A = \partial_1^C \Gamma_1.
\end{eqnarray}
The choices of initializing and measuring ancilla qubits in Z-basis is made to preserve detectors from Z-checks of $C_0$ throughout the procedure. Under such choices, one can simplify the problem of finding fault-tolerant surgery protocols to the search of merge codes with sufficient distance.
\begin{lemma}[Phenomenological distance is gauranteed by merge code distance; \cite{williamson2024lowoverheadfaulttolerantquantumcomputation}] \label{lem:merge_dist_imply_ph}
    Following the described surgery protocol, a sufficient merge code distance, $d_{\text{merge}} \geq d$, is sufficient and necessary condition for 
    \begin{eqnarray}
            d_{\text{ph}}\geq d
    \end{eqnarray}
    where $d_{\text{ph}}$ is the phenomenological distance.
\end{lemma}
We refer interested readers to Ref.~\cite{williamson2024lowoverheadfaulttolerantquantumcomputation} for a complete proof, but we remark that this is nontrivial since in general the merge code distance is only a necessary condition for the merged code distance. To gaurantee the phenomenological distance, one needs to consider the subsystem spacetime code distance~\cite{Vuillot_2019}. The proof of the dressed distance comes from the careful preservation of detectors following the designed procedures outlined above. With Lemma~\ref{lem:merge_dist_imply_ph}, we implicitly assume the above procedure is followed in later discussions and focus on the merged code distance.

\begin{remark}
        To gaurantee fault-tolerance, the merge code should be LDPC as well~\cite{williamson2024lowoverheadfaulttolerantquantumcomputation}. Thus, the additional boundaries and chain maps, including $\partial_{1,2}^{A,C}$ and $\Gamma_{0,1}$, should be sparse. This is in general not gauranteed and where the challenge comes from. In Appendix~\ref{app:structered_ancilla}, we present constructive approaches to construct these sparse connections. In other sections, we assume this desiderata is satisfied.
\end{remark}

\subsection{Measurements through metachecks}

While the ancilla chain complex never exists by itself during surgery, its property is indicative and sufficient for characterizing the effect of the surgery. To start with, we add an additional vector space, $A_2$, to $\mathcal{A}$ known as the metacheck~\cite{Campbell_2019} such that $\mathcal{A}: A_2 \to A_1\to A_0\to A_{-1}$. The X-metachecks characterizes the deterministic outcomes in combination of X-checks, i.e. the redundancies in the X-checks. While there is a degree of freedom to choose $A_2$ as long as the boundary condition is satisfied, we will choose it such that $\mathcal{A}$ is exact at $A_1$ by default.

The criticality of metachecks in $A_2$ rises from that we follow the surgery protocol where $A_0$ are initialized in Z-basis. As a result, by the stabilizer update rule, any X-stabilizer that has support in $A_0$ has random outcome at the initial merge step, providing no information on data code by itself. $\im{\partial^A_2}$ exactly represents the combination of X-checks that have no support in $A_0$. To understand what information one is measuring, we can draw the following commuting diagram,
    \begin{eqnarray}\label{eq:long_chain_diag}
        \begin{tikzcd}[column sep=large, row sep=large]
        \mathcal{A}\text{:} & {A_2} & {A_1} & {A_0} & {A_{-1}}\\
         & {} & {C_2} & & \\
        \mathcal{C}^\prime\text{:} & {}& {\oplus }& {C_1} & {C_0} \\
        & & {C_2^L} & & 
        \arrow["{\partial^A_2}", from=1-2, to=1-3]
        \arrow["{\partial^A_1}", from=1-3, to=1-4]
        \arrow["{\partial^A_0}", from=1-4, to=1-5]
        \arrow["{\partial^C_2}", from=2-3, to=3-4]
        \arrow["{\partial^C_1}", from=3-4, to=3-5]
        \arrow["{\Gamma^S_2}", from=1-2, to=2-3]
        \arrow["{\Gamma^L_2}", from=1-2, to=4-3]
        \arrow["{\partial_2^{C, L}}", from=4-3, to=3-4]
        \arrow["{\Gamma_1}", from=1-3, to=3-4]
        \arrow["{\Gamma_0}", from=1-4, to=3-5]
        % \arrow["{\gamma_0}", from=1-3, to=2-3]
\end{tikzcd}
\end{eqnarray}
where $\partial_2^{C,L} = \sbkt{L_X^C}^T$. We have added $C_2^L$ to represent the space of X-logicals so that the metacheck mapping relation is clear. The addition of $C_2^L$ turn $\mathcal{C}^\prime$ exact at $C_1$. The existence of $\Gamma_2^{S, L}$ are then gauranteed by the exactness of both chain complex and Lemma~\ref{lem:existence_com_diag}. 

It is now straightforward to understand why such metachecks enable the measurement of large-weighted X-logicals in $C_2^L$: instead of using the metachecks for correcting measurement errors in single-shot error correction~\cite{Campbell_2019}, they are used to infer the values of $C_2^L$ through $\Gamma_2^L$ and measured parity checks of $A_1, C_2$. More concreteley, the measured operators are characterized through the following lemma, which is a generalization of Lemma 1 in Ref.~\cite{Ide_2025}. 
\begin{lemma}[Surgery through homological measurements]\label{lem:surgery_measurements_framework}
    Consider a surgery procedure, the space of all measured operators is given by
    \begin{eqnarray}
            M = \Gamma_1\sbkt{\ker{\partial_1^A}} \subseteq S_X^C \oplus L_X^C
    \end{eqnarray}
    and the measured logicals is given by $M_L := M\cap L_X^C$. 
\end{lemma}
Recall that the central idea of surgery schemes is to weight-reduce weight-$\Omega\sbkt{d}$ logicals into weight-$\bigO{1}$ stabilizers. Under the reasonable assumptions $\partial_2^C, \Gamma_1$ in Eq.~\eqref{eq:long_chain_diag} are LDPC and that elements of $\im{\partial_2^{C,L}}$ have weight $\Omega\sbkt{d}$, a necessary requirement is $\abs{\partial_2^A} = \Omega\sbkt{d}$, i.e. metachecks of $\mathcal{A}$ should be global. This rules out the potential of adopting codes whose metachecks are generated by local metachecks, such as the gauge color code~\cite{Bomb_n_2015, Brown_2016}.

The efficiency of a surgery scheme is determined by the spacetime cost scaling per logical measurement. Most of the existing works are low-rate in this regards as their ancilla are composed of one or more disjoint ancilla systems with $\dim\sbkt{\ker{\partial_1^A}} = 1$. This is a direct consequence of the kernel space of a connected graph being one-dimensional. A more formal definition of the efficiency through information extraction rate is given below.
\begin{definition}[Information extraction rate and efficiency of surgery] Consider an ancilla system $\mathcal{A}$ and measured operators, $M$, given by Lemma~\ref{lem:surgery_measurements_framework}, the information extraction rate is defined as
\begin{eqnarray*}
    r_M := \dim\sbkt{M}/\abs{\mathcal{A}}.
\end{eqnarray*}
where $\abs{\mathcal{A}} := \sum_{i}\abs{A_i}$. A surgery scheme is low-rate if $r_M = \bigO{1/d}$ under the assumption of a $\Theta\sbkt{d}$-weight basis, and it is defined as high-rate otherwise. A constant-rate surgery scheme is one with $r_M = \Theta\sbkt{1}$.
\end{definition}
The information extraction rate and the space overhead is directly connected by the number of logical measurements. A constant-rate surgery scheme leads to constant spacetime overhead, $\alpha = \Theta\sbkt{1}$, regardless of number of measurements. While in the definition we assumed a minimum-weight logical basis for simplicity, it can be naturally generalized by consider low-rate schemes to be $r_M = \bigO{1/w}$ with $w:= \min \abs{l}$ for $l\in M_L$.

\subsection{Merge code distance}

In addition to the logical action, it is crucial to ensure the protocol is distance-preserved, which is determined by the merge code distance. In a surgery scheme that measures X-logicals following Eq.~\eqref{eq:commute_diagram_app}, the Z-distance is preserved regardless. On a high-level, the reason is that all X-checks of the data code, $C_2$ remained intact. Therefore, any dressed Z-logical's support on $C_1$ must be a Z-logical of the data code, which has sufficient distance by definition. If a Z-logical do not have support on $C_1$, it cannot be a dressed logical.
\begin{lemma}[Z-distance preservation; \cite{Cohen2022LDPCSurgery,Cross2024ImprovedSurgery,Ide_2025,williamson2024lowoverheadfaulttolerantquantumcomputation}] 
    For a surgery scheme, the merge code Z-distance is always preserved, 
    \begin{eqnarray}
            d^Z_{\text{merge}} \geq d.
    \end{eqnarray}
\end{lemma}

The X-distance is more challenging to prove. Before diving into the technical, it is important to interpret the logical space of $\mathcal{A}$. These logicals, if exist, are initialized deterministically in $\ket{0_L}^{\otimes k_A}$ and measured in Z-basis. Since they do not encode information, they are effectively gauge logicalsof the merge code. It is important to note that one generally cannot use these logicals to correct errors since they are of weight $\Omega\sbkt{d_A}$. These gauge logicals can reduce the code distance by forming dressed logicals as we defined in Appendix~\ref{app:notation}. Three remedies have been taken to treat these gauge logicals.
\begin{enumerate}
    \item Gauging~\cite{williamson2024lowoverheadfaulttolerantquantumcomputation, Ide_2025}: the gauge degrees of freedom are promoted into large-weighted stabilizers. These stabilizer are then treated with techniques from quantum weight reduction~\cite{hastings2016weightreductionquantumcodes, hastings2023quantumweightreduction, Sabo_2024}.
    \item Eliminating gauge logicals by construction: under certain assumptions on, e.g., the data code encoding and chain maps~\cite{Cross2024ImprovedSurgery}, one can provide explicit constructions of the ancilla system such that $k_A = 0$. In Methods, we outlined a more general approach where we do not place assumptions on the data code but only assume the logical basis takes on a canonical form. This approach also falls in this category.
        \item Subsystem code~\cite{Cohen2022LDPCSurgery,Guo2025TimeEfficient}: the merge code is accepted as a subsystem code where one focuses on boosting its dressed distance. A typical technique is to boost the initial ancilla system by thickening.
\end{enumerate}
The caveat of the last approach is that the gauge logicals grow as one thickens, so the required number of thickened layers cannot be bounded by the expansion of the initial ancilla system as one does when there are no gauge logicals~\cite{Cross2024ImprovedSurgery}. To measure a single logical, the first approach is asymptotically optimal up to polylogarithm factors, and it has shown decent performance in practical regimes as well~\cite{yoder2025tourgrossmodularquantum, Ide_2025}. However, it heavily relies on graph subroutines such as cellulation and decongestion~\cite{Freedman_2021}, which are not straightforwardly extendable to hypergraphs. The second approach is tightly connected to the structure of the data code and logical basis chosen. As an example, when $\mathcal{C}$ has Z-metachecks, it is possible that the logicals of $\mathcal{A}$ are in turn measured through these Z-metachecks. This is straightforward to see if one consider the cohomology picture of Eq.~\eqref{eq:commute_diagram_app}, i.e. reverting all arrows. In this dual picture, $\mathcal{C}$ is the surgery ancilla system for $\mathcal{C}$ for measuring Z-logicals or Z-stabilizers, following Lemma~\ref{lem:surgery_measurements_framework}. Therefore, some gauge logicals of $\mathcal{A}$ might be naturally elimiated. This is first proposed in Ref.~\cite{Cross2024ImprovedSurgery} for when $\mathcal{C}$ is a HGP codes. While this is not uncommon to occur, we wish to focus on general data codes in this section, so we do not make similar assumptions. Therefore, in this work we mainly consider the last approach and work with subsystem codes with the same gauge logical number as $k_A$.

A general sufficient condition for the X-distance is most conveniently expressed with the guaged ancilla code and commuting diagram as following
    \begin{eqnarray}\label{eq:gauge_commute_diag}
        \begin{tikzcd}[column sep=large, row sep=large]
        {A^g_1} & {A_0} & {A_{-1}} \\
        {C_2} & {C_1} & {C_0}
        \arrow["{\partial^{A, g}_1}", from=1-1, to=1-2]
        \arrow["{\partial^A_0}", from=1-2, to=1-3]
        \arrow["{\partial^C_2}", from=2-1, to=2-2]
        \arrow["{\partial^C_1}", from=2-2, to=2-3]
        \arrow["{\Gamma_{1}^g}", from=1-1, to=2-2]
        \arrow["{\Gamma_0}", from=1-2, to=2-3]
        % \arrow["{\gamma_0}", from=1-3, to=2-3]
\end{tikzcd}
\end{eqnarray}
where $A^g_1 = A_1\oplus A_1^L$, and the new boundaries are 
\begin{eqnarray}
        \partial_1^{A,g} = \begin{pmatrix}
            \partial_1^A & \sbkt{G_A^X}^T
        \end{pmatrix}, 
        \Gamma_1^g = \begin{pmatrix}
            \Gamma_1 & \Gamma_g
        \end{pmatrix}.
\end{eqnarray}
Here, $G_A^X$ is a generating basis for the gauge X-logicals in $A_1^L$, and $\Gamma_g$ is an induced chain map. The existence of this commuting diagram is similarly gauranteed by the exactness of chain complex, Lemma~\ref{lem:existence_com_diag}, and the assumption that $\mathcal{C}$ has no Z-metachecks. Based on this gauged system, we arrive at a sufficient condition for general X-distance. To arrive at the concise lemmas, we relax the inequalities in multiple steps. We comment at the end of this section how these overlooked structures can be crucial and utilized. We also remark on how our result connects to prior works, which are based on (relative) Cheeger constants~\cite{Cross2024ImprovedSurgery, swaroop2025universaladaptersquantumldpc}.
\begin{theorem}[X-distance] \label{thm:general_X_dist}
    Consider the gauged boundary $\partial_1^{A,g}$ and the induced commuting diagram in Eq.~\eqref{eq:gauge_commute_diag}. The merge code X-distance satsify
    \begin{eqnarray}
            d^X_{\text{merge}} \geq \min\sbkt{1, \frac{\rho}{\abs{\Gamma_1^g}}} d.
    \end{eqnarray}
     if $\partial_1^{A,g}$ has soundness $\rho$ up to $t\geq d$,
\end{theorem}
\begin{proof}
    The distance proof is a generalization of past works but the underlying idea is very similar. An observation is that $\forall l_C^X \in L_C^X$, it stills commutes with all Z-checks in the merge code since the merge step does not apply additional Z-checks on $C_1$. Therefore, a bare logicals basis is given by $L_C^X\setminus M_L$, where $M_L$ are the measured logicals such that $M_L\in S^X_{\text{merge}}$.

    Given a bare logicals basis, the weight can be reduced by any combination of stabilizers and gauge logicals. Before merging, the gauge logicals have a generating basis of $G_A^X$ by definition. After merging, the gauge logicals' supports extends to $A_0\oplus C_1$ with a generating basis of $\begin{pmatrix}
        \sbkt{G_A^X}^T\\ \Gamma_g
    \end{pmatrix}$. One can straightforwardly see that such a basis is not in the stabilizer group, $S^X_{\text{merge}}$, and it commutes with all Z-checks since Eq.~\eqref{eq:gauge_commute_diag} forms a commuting diagram. The ancilla X-stabilizers are similarly generated by a basis of $\begin{pmatrix}
        \partial_1^A\\ \Gamma_1
    \end{pmatrix}$. 
    Therefore, dressed logicals in the merged code have the following form. $\forall l^X_{\text{merge}}\in L_{\text{merge}}^X$, $\exists l^X_C\in L_C^X, x_s, x_g$ such that
    \begin{eqnarray}
            l^X_{\text{merge}} = \begin{pmatrix}
                l_C^X\\ 0
            \end{pmatrix} + \begin{pmatrix}
        \partial_1^A\\ \Gamma_1
    \end{pmatrix} x_s + \begin{pmatrix}
        \sbkt{G_A^X}^T\\ \Gamma_g
    \end{pmatrix} x_g
    \end{eqnarray}
    Therefore, the dressed logical weight is 
        \begin{eqnarray}
             \abs{l^X_{\text{merge}}} &=& \underbrace{\abs{l_C^X + \Gamma_1 x_s + \Gamma_g x_g}}_{\text{Supp. on }C_1} + \underbrace{\abs{ \partial_1^A x_s + \sbkt{G_A^X}^T x_g }}_{\text{Supp. on }A_0}\\
            &=& \abs{l_C^X + \Gamma_1^g x} + \abs{\partial_1^{A, g} x}
    \end{eqnarray}
    where we define $x := \begin{pmatrix}
        x_s\\ x_g
    \end{pmatrix}$. Here we do not need to explicitly include stabilizers of $\mathcal{C}$ because we accept any logical representation of $l_C^X$. It is important to observe that this logical weight and the following derivation is equivalent with respect to any kernel vectors $u\in \ker{\partial_1^{A,g}}$: let $x = x^\prime + u$, we have 
            \begin{eqnarray}
             \abs{l^X_{\text{merge}}} &=& \abs{l_C^X + \Gamma_1^g x} + \abs{\partial_1^{A, g} x} = \abs{l_C^{X, \prime} + \Gamma_1^g x^\prime} + \abs{\partial_1^{A, g} x^\prime}
    \end{eqnarray}
    where $l_C^{X, \prime} = l_C^{X} + \Gamma_1^g u$. Note that the gauge logicals generators are by default an independent basis, and it is not in the image of $\partial_1^A$. Thus, $\ker{\partial_1^{A,g}} = \set{\begin{pmatrix}
        u_s \\ 0
    \end{pmatrix} \vert u_s \in \ker{\partial_1^A}}$. Through Lemma~\ref{lem:surgery_measurements_framework}, we know $\Gamma_1\sbkt{\ker{\partial_1^A}}\subseteq S_X^C \oplus L_X^C$. Thus, $l_C^{X, \prime} \in L_X^C$. Since $\abs{l_C^{X}}, \abs{l_C^{X, \prime}}\geq d$ by definition of distance, to prove $ \abs{l^X_{\text{merge}}}\geq d$ it is sufficient to discuss when $\abs{\Gamma_1^g x} = d_{\Gamma_1^g}\sbkt{x, \ker{\partial_1^{A,g}}}$. Then, we have
    \begin{eqnarray}
            \abs{l^X_{\text{merge}}} &=& \abs{l_C^X + \Gamma_1^g x} + \abs{\partial_1^{A, g} x} \\
            &\geq & \abs{\abs{l_C^X} - \abs{\Gamma_1^g x}} + \abs{\partial_1^{A, g} x}\\
            &=& \abs{\abs{l_C^X} - d_{\Gamma_1^g}\sbkt{x, \ker{\partial_1^{A,g}}}} + \abs{\partial_1^{A, g} x}
    \end{eqnarray}
    One have already, $\abs{l^X_{\text{merge}}}\geq \abs{\partial_1^{A, g} x}$, so we only require below the soundness condition to hold for $\abs{\partial_1^{A, g} x}< d$. Now we discuss case by case. 
    \begin{enumerate}
        \item When $\abs{l_C^X} < d_{\Gamma_1^g}\sbkt{x, \ker{\partial_1^{A,g}}}$, we have 
        \begin{eqnarray}
                \abs{l^X_{\text{merge}}} &> & \abs{\partial_1^{A, g} x} \geq \rho \cdot d\sbkt{x, \ker{\partial_1^{A,g}}}\geq \frac{\rho}{\abs{\Gamma_1^g}} d
        \end{eqnarray}
        where the inequalities follow from the definition of soundness and Eq.~\eqref{eq:dist_ineq}
        \item When $\abs{l_C^X} \geq d_{\Gamma_1^g}\sbkt{x, \ker{\partial_1^{A,g}}}$, we have 
        \begin{eqnarray}
                \abs{l^X_{\text{merge}}} &\geq& \abs{l_C^X} - d_{\Gamma_1^g}\sbkt{x, \ker{\partial_1^{A,g}}} + \abs{\partial_1^{A, g} x}\\
                &\geq &d - \sbkt{\abs{\Gamma_1^g}-\rho}d\sbkt{x, \ker{\partial_1^{A,g}}}
        \end{eqnarray}
        Note that to expand the absolute value, we have imposed the initial assumption and that $d - \abs{\Gamma_1^g}d\sbkt{x, \ker{\partial_1^{A,g}}}\geq 0$, leading to $d\sbkt{x, \ker{\partial_1^{A,g}}}\leq\frac{d}{\max\sbkt{\rho, \abs{\Gamma_1^g}}}$. Thus, we arrive at 
        \begin{eqnarray}
                \abs{l^X_{\text{merge}}} \geq \begin{cases}
                    \frac{\rho}{\abs{\Gamma_1^g}} d, &\text{ if } \abs{\Gamma_1^g} \geq \rho\\
                    \frac{2\rho - \abs{\Gamma_1^g}}{\rho}d,& \text{ if } \abs{\Gamma_1^g} < \rho
                \end{cases}
        \end{eqnarray}
        The second case is gauranteed to be larger than $d$.
    \end{enumerate}
    Therefore, above concludes that we can gaurantee 
    \begin{eqnarray}
            d^X_{\text{merge}} \geq \min\sbkt{1, \frac{\rho}{\abs{\Gamma_1^g}}} d.
    \end{eqnarray}
\end{proof}

While the above considers codes with linear soundness $\rho$, it is straightforward to generalize the result to good but nonlinear soundness as in Definition~\ref{def:campbell_soundness}.

\begin{corollary}[X-distance under general soundness] 
    Assume the boundary $\partial_1^{A,g}$ is $(t,f)$-sound for $t\geq d$, the merge code X-distance satsify
    \begin{eqnarray}
            d^X_{\text{merge}} \geq \min\sbkt{d, f^{-1}\sbkt{\frac{d}{\abs{\Gamma_1^g}}}}
    \end{eqnarray}
    where $f$ is assumed to be monotoically increasing.
\end{corollary}
\begin{proof}
    The proof proceeds largely the same as in Theorem~\ref{thm:general_X_dist}. The major difference is that we replace the syndrome weight constraints with 
    \begin{eqnarray}
            \abs{\partial_1^{A,g}}\geq f^{-1}\sbkt{d\sbkt{x, \ker{\partial_1^{A,g}}}} \geq f^{-1}\sbkt{\frac{1}{\abs{\Gamma_1^g}}d\sbkt{x, \ker{\partial_1^{A,g}}}}.
    \end{eqnarray}
    The invertibility of $f$ is gauranteed by it being monotonically increasing. The corollary then follows the same case-by-case discussion.
\end{proof}

The above lemma puts a constraint on the ratio between soundness and weight of $\Gamma_1^g = \begin{pmatrix}
            \Gamma_1 & \Gamma_g
        \end{pmatrix}$. However, this is not very useful in general since while $\rho = \bigO{1}$ by the constraint of LDPC $\partial_1^A$, there is no constraint on $\abs{\Gamma_g}$. The following corollary replaces the constraint on $\abs{\Gamma_g}$ by an assumption on the X-distance of $\mathcal{A}$.

\begin{corollary}[X-distance with sufficient ancilla X-distance] \label{cor:X_dist_suff_anc_dist}
    Assume the boundary $\partial_1^{A}$ has soundness $\rho$ up to $t\geq d$, the merge code X-distance satsify
    \begin{eqnarray}
            d^X_{\text{merge}} \geq \min\sbkt{1, \frac{\rho}{\abs{\Gamma_1}}} d.
    \end{eqnarray}
    if the ancilla code, $\mathcal{A}$, has sufficient X-distance, $d_A^X\geq d$.
\end{corollary}
\begin{proof}
    The proof extends from Theorem~\ref{thm:general_X_dist}. The key difference is that $\forall x\in A_1^g$, we have
    \begin{eqnarray}
            \partial_1^{A,g}x = \begin{pmatrix}
            \partial_1^A & \sbkt{G_A^X}^T
        \end{pmatrix}\begin{pmatrix}
            x_{A_1} \\ x_{g}
        \end{pmatrix} = \partial_1^A x_{A_1} + \sbkt{G_A^X}^T x_g
    \end{eqnarray}
    where we have assumed $x = \begin{pmatrix}
        x_{A_1}\\ x_g
    \end{pmatrix}$ with $x_{A_1}\in A_1, x_g\in A_1^L$. Consider the case where $\partial_1^{A,g}x$ is equivalent to a gauge logical up to stabilizer group, i.e. $x_g\neq 0$. The gauge X-logicals are the X-logicals of the ancilla code $\mathcal{A}$ viewed independently. Therefore, under the assumption $d_A^X\geq d$, we have by definition
    \begin{eqnarray}
            \abs{\partial_1^{A, g} x}\geq d.
    \end{eqnarray}
    Therefore, we no longer need to consider dressed logicals, i.e. one only needs to prove the minimum weight restricted to $x = \begin{pmatrix}
        x_{A_1} \\ 0
    \end{pmatrix}$. This effectively leads us back to the scenario where the code $\mathcal{A}$ has no gauge logicals, so the distance can be gauranteed by the expansion of the added X-checks. The corollary then follows analagously.
\end{proof}

Notably, under the assumption that $\abs{\Gamma_1}$ is LDPC, Corollary~\ref{cor:X_dist_suff_anc_dist} is built upon conditions solely on the independent ancilla code. From a high-level, our condition presents a general desiderata sufficient for distance-preservation: the ancilla code should have sufficient X-distance, and the transpose of its X-check should have constant classical soundness. The condition on ancilla distance is not very restrictive since it is only single-sided. This enable the opportunity to consider codes with gauranteed distance in one side only, which is natural in code constructions such as lifted product codes~\cite{Panteleev_2022, ZhengHan_in_prep}, fiber bundle codes~\cite{Hastings_2021}, and others~\cite{Evra_2020, Kaufman_2021} before distance balancing. Moreover, the search for expanding classical and quantum codes is an active direction of research~\cite{dinur2021locallytestablecodesconstant, Panteleev_2022_asymptotically, Lin_2022, Wills_2025_IEEE, Cross2024quantumlocally}. We highlight that the sufficient condition derived is distinct from the desiderata for conventional expanding codes and could motivate future research directions.

One should note that Corollary~\ref{cor:X_dist_suff_anc_dist} is a sufficient condition with relaxations. An example is that we have applied the inequality in Eq.~\eqref{eq:dist_ineq} multiple times. In fact, in prior works and some constructions we present in Appendix~\ref{app:structered_ancilla}, $\Gamma_1$ only has support on a limited subset of $A_1$. In these cases, it is unnecessary to require global expansion and local expansion is sufficient. When one narrow down to considering low-rate surgery and transversal-type chain maps $\Gamma_1 = \begin{pmatrix}
        I & \mathbf{0}\\
        \mathbf{0} & \mathbf{0}
\end{pmatrix}$, this condition is equivalent to the relative Cheeger constant condition in Ref.~\cite{swaroop2025universaladaptersquantumldpc}.

Another piece of information that could greatly relax our conditions is the structure of the logicals. Here, structure refers to both structure in the logical basis and in the chosen set of logical measurements, $M$. In many cases, especially when there is translation symmetry or product structures, these information reduces the overhead greatly. Some well-known examples include logical measurements of surface code and toric codes through lattice surgery~\cite{Horsman_2012}, which require almost minimal ancilla overhead. The ancilla system does not have sufficient global or local expansion, but it is still distance-preserving. In Appendix~\ref{app:structered_ancilla}, we utilize such structures in the data code to demonstrate examples when the requirement on the expansion property in Corollary~\ref{cor:X_dist_suff_anc_dist} can be completely removed.

\section{High-rate surgery through structured ancilla \label{app:structered_ancilla}}

In this section, we utilize the structure in the ancilla and/or the data code to (1) reduce overhead and (2) provide more concrete and constructive schemes for high-rate surgery. In Appendix~\ref{app:HGP_ancilla_general_data}, we specify the ancilla construction as a HGP code, which is an example where the ancilla code satisfy the desiderata listed in Corollary~\ref{cor:X_dist_suff_anc_dist} while being high-rate. In Appendix~\ref{app:HGP_ancilla_HGP_data}, we further specify the memory to be encoded in a HGP code. Such assumption largely alleviates the requirement on expansion, which is often argued in the asymptotic limit.

While our discussions focus on HGP codes, a large portion of them can be applicable to other codes, specifically those that bears a similar product structure, such as LP code~\cite{Panteleev_2022}, fiber bundle code~\cite{Hastings_2021}, and balanced product codes~\cite{Breuckmann_2021}. The key reason we perform the analysis on HGP codes is that its distance is well understood and possess a known complete logical basis, both of which are lacking in the other codes with better code parameters in general. We present some progress in this direction in Ref.~\cite{ZhengHan_in_prep} and leave the adaptation of discussions in this sections to those codes to future work.

\subsection{Surgery on general data code \label{app:HGP_ancilla_general_data}}
We begin with quoting some well-known results from classical coding theory, which constructs expander codes from left-expanding bipartite graphs. We only outline the most relevant result here for simplicity.
\begin{definition}[Expander codes; \cite{SipserSpielman1996ExpanderCodes, Leverrier_2015}]\label{def:expander_code}
    A $\sbkt{\gamma, \delta}$-expanding code with $\delta<1/2$ is a classical LDPC code with parameters $[n, k = \Theta\sbkt{n}, d\geq \gamma n]$. For any subset of $S\subseteq F_2^{n}$ such that $\abs{S}\leq \gamma n$, we have 
    \begin{eqnarray}
            \Gamma_u\sbkt{S} \geq (1-2\delta) \Delta \abs{S}
    \end{eqnarray}
    where $\Delta$ is the column weight of its biregular parity check matrix. $\Gamma_u\sbkt{S}$ denotes the number of checks only connected to one qubit in $S$.
\end{definition}

Next, we formally define the product construction. In particular, we first define the product complex and the hypergraph product code~\cite{Bravyi_2014_homological, PhysRevLett.122.230501, Tillich2014HGP, PRXQuantum.2.040101}.
\begin{definition}[Product complex and hypergraph product code]
    Consider two chain complex, $\mathcal{B}, \mathcal{D}$, with boundaries $\partial^{B, D}$, the product complex is defined as $\mathcal{P} = \text{Prod}\sbkt{\mathcal{B}, \mathcal{D}}$ with 
    \begin{eqnarray}
            P_k := \bigoplus_{i+j=k} B_i\otimes D_j.
    \end{eqnarray}
   $\forall b_i\otimes d_j\in P_k$, the boundary maps act as $\partial^P_k \sbkt{b_i\otimes d_j} = \partial_i^B b_{i}\otimes d_j +  b_{i}\otimes \partial_j^D d_j$. A (2-dimensional) hypergraph product code (HGP) is a length-2 complex given by $\text{HGP}\sbkt{\mathcal{B}, \mathcal{D}} := \text{Prod}\sbkt{\mathcal{B}, \mathcal{D}}$ with $\mathcal{B}, \mathcal{D}$ being length-1 complexes, i.e. classical codes.
\end{definition}
The definition of product complex and HGP can be naturally extended to higher dimensions to homological product codes~\cite{Bravyi_2014_homological, PhysRevLett.122.230501}, but we restrict our discussion to HGP and discuss high-dimensional extensions when possible. For better visualization, $\text{HGP}\sbkt{\mathcal{B}, \mathcal{D}}$ can be written as 
\begin{eqnarray}
        \begin{tikzcd}[column sep=large, row sep=large]
        {} & {B_1\otimes D_0} & {} \\
        {B_1\otimes D_1} &  & {B_0\otimes D_0}  \\
        {} & {B_0\otimes D_1} & 
        \arrow["I_{B_1}\otimes \partial^D", from=2-1, to=1-2]
        \arrow["\partial^B \otimes I_{D_1}", from=2-1, to=3-2]
        \arrow["I_{B_0}\otimes \partial^D", from=3-2, to=2-3]
        \arrow["\partial^B \otimes I_{D_0}", from=1-2, to=2-3]
\end{tikzcd}
\end{eqnarray}
Through Kunneth theorem, it has parameters
\begin{eqnarray}
        [[n = n_{B_{1}}n_{D_0}+ n_{B_{0}}n_{D_1}, k = k_B^T k_D + k_B k_D^T, d_X = \min\sbkt{d_B, d_D}, d_Z = \min\sbkt{d_B^T, d_D^T}]]
\end{eqnarray}
where we have defined $n_{V}:= \dim{V}$ for any vector space $V$. Distance $d$ and kernel dimension $k$ with transpose superscript represent the corresponding property of the transpose code, i.e. considering $\partial_{B, D}^T$ as the parity check matrix. Note that our definition of HGP code is different from some literature where their definitions correspond to our $\text{HGP}\sbkt{\mathcal{B}, \mathcal{D}^T}$. However, adopting our notation will make notations simpler in this work.

Consider a HGP code, $\mathcal{P}=\text{HGP}\sbkt{\mathcal{B}, \mathcal{D}}$, and restrict to the subcomplex, $P_2\longrightarrow P_1$. This exactly corresponds to a tensor product code in classical codes~\cite{bensasson2004robustlocallytestablecodes, kalachev2023twosidedrobustlytestablecodes}. We denote a complex with this structure as $\text{Tensor}\sbkt{\mathcal{B}, \mathcal{D}}$. While tensor product codes are well-known examples of robustly testable codes under certain assumptions~\cite{10.1007/11830924_29}, we derive the following expansion property for our purpose.

\begin{lemma}[Soundness properties of tensor product codes]\label{lem:soundness_tensor_code}
    Consider a tensor product code, $\text{Tensor}\sbkt{\mathcal{B}, \mathcal{D}}$, its parity check is $\sbkt{t, f}$-sound with $f(x) = x^2/4$ and $t = \min\sbkt{d_B,d_D}$~\cite{Campbell_2019}. Under the assumption that at least one of $\mathcal{B}, \mathcal{D}$ is a $(\gamma,\delta)$-expanding code, the parity check of $\text{Tensor}\sbkt{\mathcal{B}, \mathcal{D}}$ has linear soundness of 
    \begin{eqnarray}
            \rho = \sbkt{1-2\delta}\Delta
    \end{eqnarray}
    up to $t = \gamma n_{D_1}$, where $n_{D_1}$ denotes the number of bits in $\mathcal{D}$ and $\Delta$ is the weight of $\partial^D$.
\end{lemma}
\begin{proof}
    The first part of the lemma was proved in Ref.~\cite{Campbell_2019} to show single-shot properties of two- and high-dimensional hypergraph product codes. We provide a rough sketch of the idea here. The bits of a tensor code can be equivalently considered to be in $F_2^{n_{B_1}\otimes n_{D_1}}$ up to rearrangement. In this representation, the bits form a matrix of size $n_{B_1}\times n_{D_1}$. The checks are a direct sum of row checks and column checks. Consider an error configuration $e\in F_2^{n_{B_1}\otimes n_{D_1}}$, with equivalent row and column vector forms
    \begin{eqnarray}
            e = \begin{pmatrix}
                c_1 & c_2 & \dots & c_{n_{B_1}}
            \end{pmatrix} = \begin{pmatrix}
                r^T_1 \\ r^T_2 \\ \vdots \\ r^T_{n_{A_1}}
            \end{pmatrix}.
    \end{eqnarray}
    Suppose the boundary operator of $\text{Tensor}\sbkt{\mathcal{B}, \mathcal{D}}$ is $\partial$. The syndrome is given by
    \begin{eqnarray}
            \partial e = \begin{pmatrix}
                \partial^B c_1 & \partial^B c_2 & \dots & \partial^B c_{n_{B_1}}
            \end{pmatrix} \oplus \begin{pmatrix}
                \partial^D r_1 \\ \partial^D r_2 \\ \vdots \\ \partial^D r_{n_{A_1}}
            \end{pmatrix}
    \end{eqnarray}
    A case-by-case discussion follows as below.
    \begin{enumerate}
        \item If none of the columns and rows of $e$ are in $\ker{\partial^{B, D}}$ respectively, the syndrome size is lower bounded by the sum of its column and row support size, while the error weight is upper bounded by the product of its column and row support size. This leads to 
        \begin{eqnarray}
            \abs{e} \leq \abs{\text{rowsupp}(e)} \abs{\text{rowsupp}(e)} \leq \frac{1}{2}\sqrt{\abs{\text{rowsupp}(e)} + \abs{\text{colsupp}(e)}} \leq \frac{1}{2}\sqrt{\abs{\partial e}},
        \end{eqnarray}
        \item If certain columns of $e$ are in $\ker{\partial^{B}}$ but none of its rows are in $\ker{\partial^{D}}$, since the columns are of weight at least $d_B$, the row syndrome size is lower bounded by $d_B$. The symmetric case where rows of $e$ are in $\ker{\partial^D}$ is argued similarly.
        \item If there exists a pair of index $(i,j)$ with $i\in [n_{B_1}], j\in [n_{A_1}]$ such that both $r_i\in \ker{\partial^A}$ and $c_j\in \ker{\partial^B}$, we can consider an equivalent configuration $e^\prime = e + r_i\otimes c_j$. Considering these two error configurations are equivalent under soundness analysis because it is only concerned with the reduced distance, $d\sbkt{e, \ker{\partial}}$ with $\ker{\partial} = \ker{\partial^A}\times \ker{\partial^B}$. After applying one such transformation for each paired index $(i,j)$, $e^\prime$ will reduce to either case (1) or (2).
    \end{enumerate}
    The conclusion from the above three cases lead to the result presented in the first half of the lemma.

    If additional expansion assumption is made on either of the codes, the soundness is boosted to linear scaling. Without loss of generality, consider $\mathcal{D}$ to be a $(\gamma,\delta)$-expanding code as defined in Definition~\ref{def:expander_code}. Previous discussion in case 2 and 3 follow identically. Under the assumption of case 1, we note that
    \begin{eqnarray}
            \abs{\partial e}\geq\abs{\text{rowsupp}(e)}\geq \max_i \abs{r_i}.
    \end{eqnarray}
    Thus, if we restrict to considering $\forall r_i, \abs{r_i}\leq t = \gamma n_{B_1}$, the error weight is upper bounded as $\partial^B r_i \geq \sbkt{1-2\delta}\Delta \abs{r_i}$, which leads to
    \begin{eqnarray}
            \abs{\partial e} \geq \abs{\begin{pmatrix}
                \partial^B r_1 \\ \partial^B r_2 \\ \vdots \\ \partial^B r_{n_{A_1}}
            \end{pmatrix}} \geq \sbkt{1-2\delta}\Delta \abs{e}.
    \end{eqnarray}
    Here, we have used $\abs{e} = \sum_i \abs{r_i}$. This inequality immediately implies linear soundness with
    \begin{eqnarray}
            \rho = \sbkt{1-2\delta}\Delta
    \end{eqnarray}
    up to $t$.
\end{proof}
The above construction might seem similar to quantum expander codes~\cite{Leverrier_2015}, which are known to be confined but not sound. However, the difference lies in that their construction, in our notation, is $\text{HGP}\sbkt{\mathcal{B}, \mathcal{D}}$ with $\mathcal{B}, \mathcal{D}^T$ both based on expander graphs. Moreover, the expansion property of interest to them are the $X,Z$ check matrices, while our soundness condition is on the transpose of the X check matrix only.

\begin{theorem}[HGP ancilla with expander code is distance-preserving]
    Consider a surgery procedure equipped with a LDPC commuting diagram and a data code, $\mathcal{C}$, with distance $d$. Let the ancilla code $\mathcal{A} = \text{HGP}\sbkt{\mathcal{B}, \mathcal{D}}$, where the classical codes have parameters $[n_{\alpha}, k_{\alpha}, d_{\alpha}\geq d]$ with $\alpha = B, D$. Suppose $\mathcal{D}$ is a $(\gamma,\delta)$-expanding code with soundness $\rho = \sbkt{1-2\delta}\Delta$ up to $\gamma n_{D_1}$, where $n_{D_1}$ denotes the number of bits in $\mathcal{D}$ and $\Delta$ is the weight of $\partial^D$. The merge code distance is 
    \begin{eqnarray}
            d_{\text{merge}} \geq \min\sbkt{d, \frac{\sbkt{1-2\delta}\Delta}{\abs{\Gamma_1}}d, \gamma n_B} = \Theta\sbkt{d}
    \end{eqnarray}
\end{theorem}
\begin{proof}
    The proof is a straightforward combination of Lemma~\ref{lem:soundness_tensor_code} and Corollary~\ref{cor:X_dist_suff_anc_dist}.
\end{proof}

By far, we have derived the conditions for distance-preservation with a HGP ancilla. However, it is based on the assumption that there exists a LDPC commuting diagram composed of $\Gamma_{0,1}$ given a data code $\mathcal{C}$ and logical measurements $M$. It is nontrivial to construct such diagrams. In the following, we provide examples on how to construct these diagrams and what logical measurements can be performed in these cases. With a HGP ancilla code, $\mathcal{A}:=\text{Tensor}\sbkt{\mathcal{B}, \mathcal{D}}$, the X metachecks are given by the codewords of the tensor product code, i.e.
\begin{eqnarray}
        \ker{\partial^A_1} = \ker{\partial^B}\otimes \ker{\partial^D}.
\end{eqnarray}
An example of structured logicals one can measure is when the following condition holds.
\begin{definition}[Logical operators with generalized parallel structure]\label{def:gen_par_structure}
    Suppose the set of logical operators to measure is $M$. A partition of $M$ is collection of subsets $M_i$ with $i\in[l]$ satisfying $M_i\subseteq M$ and $\bigcup_i M_i = M$. We say $M$ possess a generalized parallel structure if $\exists \mathcal{B}, M_i, \gamma^i_1$ such that $\forall i\in[l]$
    \begin{eqnarray}\label{eq:gen_par_structure}
            \gamma^i_1\sbkt{\ker{\partial^B}} = \spanv{M_i} 
    \end{eqnarray}
\end{definition}
Recall that if Eq.~\eqref{eq:gen_par_structure} holds, a commuting diagram follows from Lemma~\ref{lem:existence_com_diag}, i.e. $\exists \gamma_0^i$ with
\begin{eqnarray}
        \partial_1^C\gamma^i_1  = \gamma^i_0 \partial^B.
\end{eqnarray}
Given the maps $\gamma_{0,1}^i$, one can construct the commuting diagram as 
\begin{eqnarray}\label{eq:example_connectivity_HGP}
        \Gamma_{0,1} = \sum_{i\leq k_D} \gamma_{0,1}^i \otimes \sbkt{e^D_i}^T
\end{eqnarray}
where $e^D_i \in D_1$ is a unit vector with its nonzero element on the $i$-th information bit of code $\mathcal{D}$. One can verify with Lemma~\ref{lem:surgery_measurements_framework} that this connection indeed measures the full set $M$ provided that the number of subsets is upper bounded as $l\leq k_D$. The above construction utilize the property that each logical of $\mathcal{D}$ in the canonical form have a unique support on its information bits. Adopting the canonical basis here is only for the purpose of extracting logical information and does not pose restrictions on $\gamma_{0,1}^i$ or the data code.

The above structure is generalized compared to strictly parallel structures. Those strictly parallel structures have been explored previously in the context of homomorphic measurements~\cite{Xu2025Homomorphic} under the name of grid PPMs (GPPMs) for the measurement of a HGP data code with a HGP ancilla code. We give more detailed comparisons in Appendix~\ref{app:HGP_ancilla_HGP_data} when we narrow down to HGP data codes.

To analyze the overhead more systematically, we consider the optimal and worst case scenario. Since $\mathcal{D}$ is asymptotically good and can be chosen with parameters $[n_D = \Theta\sbkt{d}, k_D = \Theta\sbkt{d}, d_D = d]$, the spacetime overhead of measuring $t$ logicals is generally given by $\alpha = \bigO{\frac{n_C + n_B n_D}{k_C + t} \frac{t}{k_B l}}$. Notice that $n_B\in \mbkt{d, n_C}$, depending on the choice of logical basis and the specific logical measurements. In the optimal scenario, the logicals can be partitioned such that each subset $\abs{M_i} = \Theta\sbkt{n_B} = \Theta\sbkt{d_C}$,
\begin{eqnarray}\label{eq:opt_st_overhead}
        \alpha_{\text{opt}} = \bigO{\frac{d^2_C}{n_C}}
\end{eqnarray}
where we have assumed $t = \Theta\sbkt{k_B l} = \bigO{k_C}$ and $k_C = \Theta\sbkt{n_C}$ for simplicity. We will see later that this holds for HGP data codes. In the worst scenario, we have $l=1$ and $k_B = t$. $\mathcal{D}$ can be chosen as a repetition code with sufficient distance, which leads to an overhead similar to the devised sticker approach in Ref.~\cite{Guo2025TimeEfficient}, where we generally expect $n_A\sim n_C$ and an overhead of 
\begin{eqnarray}
        \alpha_{\text{worse}} = \bigO{d_C}.
\end{eqnarray}
The above construction still requires an expanding code to form the HGP ancilla. While this is acceptable as an asymptotic argument, it is not necessarily practical for near-term implementations since the hidden constant factors might be daunting. One example where the expansion requirement can be alleviated is when $M$ only consists of a single element and we construct $\mathcal{B}$ with $k_B=1$. Here, the ancilla X-check $A_1$ is connected to the data code at $C_1$ only through one non-zero entry per row, thus the relative expansion of $\partial_1^A$ is at least 1. Therefore, merge code distance is gauranteed whenever $\mathcal{D}$ has sufficient distance. This is precisely why past schemes that rely on thickening works. However, this argument is restricted to low-rate surgery. In the following, we provide a high-rate surgery example.

\subsection{Surgery on Hypergraph Product Codes\label{app:HGP_ancilla_HGP_data}}

In the last subsection, we refrained from specifying the data code for maximum generality. In this subsection, we consider specifically the HGP code as data code to give explicit constructions. In particular, we consider the case where $\mathcal{C} = \text{HGP}\sbkt{\mathcal{B}, \mathcal{D}}$ and the ancilla code is written as $\mathcal{A} = \text{HGP}\sbkt{\mathcal{B}^\prime, \mathcal{F}}$. Then, we have the commuting diagram 
\begin{eqnarray}
        \begin{tikzcd}[column sep=large, row sep=large]
        & {} & {B^\prime_1\otimes D_0} & {} \\
       \mathcal{A}\text{:}& {B^\prime_1\otimes D_1} &  & {B^\prime_0\otimes D_0}  \\
        & {} & {B^\prime_0\otimes D_1} & \\
        & {} & {B_1\otimes F_0} & {} \\
        \mathcal{C}\text{:} & {B_1\otimes F_1} &  & {B_0\otimes F_0}  \\
        & {} & {B_0\otimes F_1} & 
        \arrow["I_{B^\prime_1}\otimes \partial^D", from=2-2, to=1-3]
        \arrow["\partial^{B^\prime} \otimes I_{D_1}", from=2-2, to=3-3]
        \arrow["I_{B^\prime_0}\otimes \partial^D", from=3-3, to=2-4]
        \arrow["\partial^{B^\prime} \otimes I_{D_0}", from=1-3, to=2-4]
        \arrow["I_{B_1}\otimes \partial^F", from=5-2, to=4-3]
        \arrow["\partial^{B} \otimes I_{F_1}", from=5-2, to=6-3]
        \arrow["I_{B_0}\otimes \partial^F", from=6-3, to=5-4]
        \arrow["\partial^{B} \otimes I_{F_0}", from=4-3, to=5-4]
        \arrow["\Gamma_1", from=2-2, to=4-3]
        \arrow["\Gamma_0", from=3-3, to=5-4]
\end{tikzcd}
\end{eqnarray}
Here, we construct the ancilla code such that $\mathcal{B}^\prime$ and $\mathcal{B}$ are connected through code homomorphism,
\begin{eqnarray}
        \begin{tikzcd}[column sep=large, row sep=large]
        {B^\prime_1} & {B^\prime_0} \\
         {B_1} & {B_0}
        \arrow["{\partial^{B^\prime}}", from=1-1, to=1-2]
        \arrow["{\partial^B}", from=2-1, to=2-2]
        \arrow["{\gamma_1}", from=1-1, to=2-1]
        \arrow["{\gamma_0}", from=1-2, to=2-2]
        % \arrow["{\gamma_0}", from=1-3, to=2-3]
\end{tikzcd}
\end{eqnarray}
where $\gamma_{0,1}$ are LDPC. This is equivalent to a general requirement of 
\begin{eqnarray}
        \gamma_0 \partial^{B^\prime} = \partial^B \gamma_1
\end{eqnarray}
and the logical action is defined through matrix $w$ with 
\begin{eqnarray}
        G_{B^\prime} \gamma_1^T = w G_{B}
\end{eqnarray}
which depends on the generator basis $G_{B, B^\prime}$ chosen. For simplicity of discussion, we consider a canonical basis of HGP, where $d_{B}$ number of X-logicals are solely supported in each column of the labelled by each information bit of $\mathcal{F}^T$. This is also a minimum weight basis as all logicals' weights are $\Theta\sbkt{d_B}$ if $\mathcal{B}$. The disjoint column support structure allows us to tailor the connectivities column-by-column. More concretely, we implement the full chain map as
\begin{eqnarray} \label{eq:hgp_chain_maps}
        \Gamma_{0,1} = \sum_{i} \gamma_{0,1} \otimes e_i^{F^T} \sbkt{e_i^{D}}^T
\end{eqnarray}
where $e_i^{F^T}$ and $e_i^{F^T}$ denote the unit vector on the $i$-th information bit on $\mathcal{D}$ and $\mathcal{F}^T$, respectively. Suppose the implementable logical actions are $W = \set{w}$, the generalized parallel structure in Definition~\ref{def:gen_par_structure} implies that we can simultaneously measure all logical set in each column if they are connected by logical actions in $W$. In other words, logicals to measure in the i-th column are simultaneously measurable if they are connected to a common logical basis, $G_{B^\prime}$, through tailored $w_i\in W$. The key is to realize the connectivity can vary from column to column. 

While our scheme has great flexibility, it remains unknown if it is possible to construct LDPC $\gamma_{0,1}$, or compile in constant-depth, for arbitrary logical codewords in each column through $w_i$. However, it is sufficient to see that the power is beyond prior works through several examples. In below we give several examples and compare with relevant literatures.
\begin{example} [Code modifications]
    Assume we work in the canonical basis of the classical code $\mathcal{B}$ in each column, the code puncturing and augmentation techniques induce classical code homomorphisms. Through these classical code homomorphisms, we can measure $M$ if its logicals form a hypergraph product pattern. Such a pattern correspond exactly to the GPPMs pattern in Ref.~\cite{Xu2025Homomorphic}.
\end{example}
If one restrict to puncturing and augmentation, our scheme's capability is equivalent to the homomorphic measurement scheme in Ref.~\cite{Xu2025Homomorphic}. There are some constant factor gains because in their scheme they require an additional ancilla patch for masking to complete the full grid PPM measurements while we do not. However, the asymptotic scaling is the same.

\begin{figure*}
    \centering
    \includegraphics[width = 0.6 \textwidth]{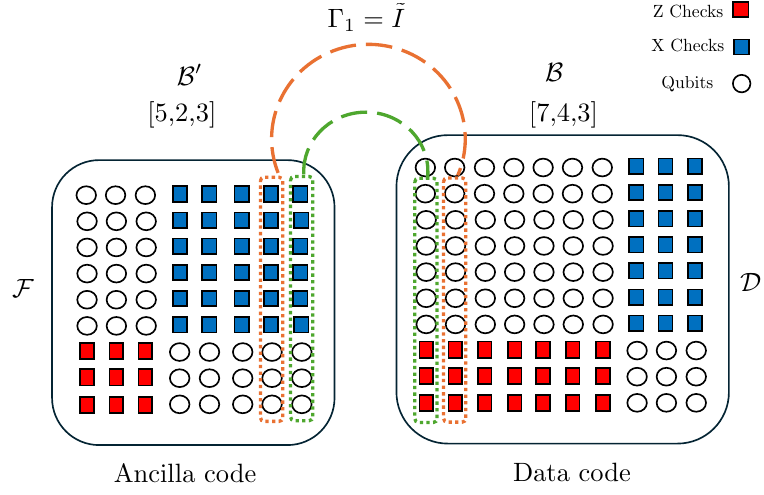}
     \caption{Illustrative figure for when the data code is a HGP code. Here the data code is formed by the hypergraph product of two $[7, 4, 3]$ Hamming codes on the right. The ancilla code, on the left, is formed by transpose checks of $[6, 3, 3]$ and $[5, 2, 3]$ on the vertical and horizontal sides respectively. The $[6, 3, 3]$ code is chosen by puncturing the Hamming code.}
    % \label{fig:algo_demo}
\end{figure*}

\begin{example} [Code automorphisms]
    Consider the automorphism group associated with $\mathcal{B}$, which is equivalent to the following conditions 
    \begin{eqnarray}
            G_B \sigma &=& wG_B\\
            \partial^B \sigma &=& v\partial^B
    \end{eqnarray}
    generated by permutations $\sigma\in S_{n_B}$ and $w\in \text{GL}_{k_B}\sbkt{\mathbb{F}_2}$. Our scheme can simulatenously measure logicals in each column if they are connected by automorphisms.
\end{example}
Quantum and classical code automorphisms have been extensively studied~\cite{malcolm2025computingefficientlyqldpccodes, berthusen2025automorphismgadgetshomologicalproduct, Grassl_2013, Sayginel2025Automorphisms,Breuckmann2024FoldTransversal}. In particular, prior works~\cite{malcolm2025computingefficientlyqldpccodes,berthusen2025automorphismgadgetshomologicalproduct} have considered explicitly how (subsystem) HGP codes inherit the classical code automorphisms, which enables efficient computation when codes with rich symmetry are adopted as base codes. For example, in subsystem HGP (sHGP) the automorphism group is inherited from classical base codes as~\cite{malcolm2025computingefficientlyqldpccodes}
\begin{eqnarray}
        \text{Aut}\sbkt{\mathcal{B}}\otimes \text{Aut}\sbkt{\mathcal{D}^T}\subseteq \text{Aut}\sbkt{\text{HGP}\sbkt{\mathcal{B}, \mathcal{D}}}
\end{eqnarray}
The choice of sHGP, at the cost of a lack of threshold, is to preserve the tensor product structure. In HGP, while the automorphism group size is similarly $\abs{\text{Aut}\sbkt{\mathcal{B}}}\cdot \abs{\text{Aut}\sbkt{\mathcal{D}}}$~\cite{berthusen2025automorphismgadgetshomologicalproduct}, not all automorphisms are gauranteed to preserve the tensor product structure. Combining the automorphism gadgets with transversal or homomorphic measurements, the number of distinct measurements one can perform scales as $\bigO{\abs{\text{Aut}\sbkt{\mathcal{B}}}\cdot \abs{\text{Aut}\sbkt{\mathcal{D}}}}$. In our scheme, since we can freely apply distinct classical automorphisms for each of the $k_D$ columns, the distinct measurements scale as $\bigO{\abs{\text{Aut}\sbkt{\mathcal{B}}}^{k_D}\cdot \abs{\text{Aut}\sbkt{\mathcal{D}}}}$.

While the presented surgery scheme can already utilize the underlying code's automorphism group more flexibly than direct inerhitence, it is unnecessarily restrictive to only consider automorphisms. One limitation of automorphism gadgets~\cite{berthusen2025automorphismgadgetshomologicalproduct, aaronson2015classificationreversiblebitoperations} is that they are not universal since they are restricted to affine logical actions, $w\in \text{GL}_{k_B}\sbkt{\mathbb{F}_2}$. However, our scheme is not limited by such rules: one can simply compose the automorphism logical actions to go beyond affine gates. For example, assume $\mathcal{B}$ has $k_B=3$ and its automorphism group include a logical permutation action $w\in \text{Aut}\sbkt{\mathcal{B}}$ such that
\begin{eqnarray}
        w = \begin{pmatrix}
            0 & 1 & 0\\
            1 & 0 & 0\\
            0 & 0 & 1
        \end{pmatrix}
\end{eqnarray}
induced by physical permutation matrices $\sigma$. We can construct chain maps by linearity as $G_B \sbkt{\sigma + I_{n_B}} = \sbkt{w + I_{k_B}}G_B$, which realizes a logical action of 
\begin{eqnarray}
        w + I_{k_B} = \begin{pmatrix}
            1 & 1 & 0\\
            1 & 1 & 0\\
            0 & 0 & 0
        \end{pmatrix}.
\end{eqnarray}
This specific logical measurement performs effectively a repeated measurement of the joint logicals, $L^X_1L^X_2$ and leaves $L^X_3$ intact. Through simple compositions as such, we arrive at logical actions no longer in the automorphism group. Moreover, this provides an example where we tailor the number of logicals measued for each column through chain maps with nontrivial kernels. Notably, if we only rely on automorphism plus classical code modifications, the number of logicals measured per column must stay the same since automorphisms are invertible.

Last but not least, we remark that HGP as data code do not require the ancilla HGP code to have an expander base code, which present it as a practical and explicit near-term code example. 
\begin{theorem}[HGP ancilla is distance-preserving for HGP data code]
    Consider a surgery procedure equipped with a LDPC commuting diagram exists and the data code is a HGP code $\mathcal{C} = \text{HGP}\sbkt{\mathcal{B}, \mathcal{F}}$ with distance $d$. Let the ancilla code be $\mathcal{A} = \text{HGP}\sbkt{\mathcal{B}^\prime, \mathcal{D}}$. If $d_D\geq d$, we have
    \begin{eqnarray}
            d_{\text{merge}} \geq \frac{1}{\abs{\Gamma_1}}d
    \end{eqnarray}
\end{theorem}
\begin{proof}
    The proof relies on the special structure of the HGP code and the way the connectivities are set. Recall that one can define the canonical basis of the classical code $\mathcal{F}^T$ and its information set, $\mathcal{I}_{F^T}$. The canonical basis for the X-logicals of the HGP $\mathcal{C}$ is spanned by logical operators of the form~\cite{Quintavalle2023partitioningqubits, Quintavalle_2022_Reshape}
    \begin{eqnarray}
            l = l^i_B\otimes e_i^{F^T},
    \end{eqnarray}
    where $l^i_B\in\ker{\partial^B}$, and $e_i^{F^T}$ are unit vectors in $\mathcal{I}_{F^T}$. The considered logical is fully contained in the $i$-th column.

    Consider the logical in canonical basis, we have
    \begin{eqnarray}
            \text{rowsupp}\sbkt{l} \geq d
    \end{eqnarray}
    and by the form of Eq.~\eqref{eq:hgp_chain_maps}, the column $i$ in $B_0\otimes F_1$ is connected to only one column in $B_1^\prime\otimes D_1$. Key is to note that given such connection, suppose we consider a stabilizer combination $x\in B_1^\prime \otimes D_1$. The logical is now
    \begin{eqnarray}
            \abs{l^\prime} &\geq& \abs{l + \Gamma_1 x} + \abs{\sbkt{I_{B_1^\prime}\otimes \partial^D} x}\\
            &\geq& \abs{\text{rowsupp}\sbkt{l} - \abs{\Gamma_1}\text{rowsupp}\sbkt{x}} + \text{rowsupp}\sbkt{x}\\
            &\geq & \frac{1}{\abs{\Gamma_1}} d
    \end{eqnarray}
    where we did not use any expansion property. All that is assumed is that for any row of x not in kernel space of $\partial^D$, the incurred qubit support is non-zero $\abs{\partial^D x\vert_{j}}\geq 1$. 

    To prove distance, one would need to consider combinations of any logical with the stabilizer group. One can first notice that considering joint paulis from different columns would not decrease weight since they are fully disjoint. The X-check matrix, restricted to $B_0\otimes F_1$, is written as 
    \begin{eqnarray}
            H_X\vert_{B_0\otimes F_1} = e^j \otimes \partial^{F^T}
    \end{eqnarray}
    where in the canonical basis 
    \begin{eqnarray}
            \partial^{F^T} = \begin{pmatrix}
                A & I^\prime
            \end{pmatrix}
    \end{eqnarray}
    with $I^\prime$ being an identity matrix of size $n_{F^T}-k_{F^T}$. Therefore, if one add a combination of checks, $s$, to row $j$ of the logical, it will incur support on the $n_{F^T}-k_{F^T}$ qubits that are the complementary of the information set. The weights of these qubits cannot be reduced by stabilizers in $B_0\otimes F_1$ since $\im{\Gamma_1}\cap \sbkt{\mathcal{I}_{F^T}}^c = \varnothing$.
\end{proof}

Consider a HGP code with good base codes such that the data code parameters are $[[n_C, k_C =\Theta\sbkt{n_C}, d_C = \Theta\sbkt{\sqrt{n_C}}]]$. If we adopt the canonical basis, for each simultaneously measurable logical set defined by Definition~\ref{def:gen_par_structure}, the logicals partitions in each column occupy at most $\bigO{d_C}$ space. Therefore, the spacetime overhead is always in the optimal scenario described in Eq.~\eqref{eq:opt_st_overhead}. 
\begin{corollary}[Spacetime overhead of surgery on HGP data code]
    The spacetime overhead of high-rate surgery on HGP data code is 
    \begin{eqnarray}
            \alpha = \bigO{1}
    \end{eqnarray}
    for measurement of $t =\bigO{k_C}$ logicals that form a generalized parallel pattern.
\end{corollary}

\section{Algorithm design for randomized ancilla construction and simulation details \label{app:random_algorithm}}

\begin{figure*}
    \centering
    \includegraphics[width = 1.0 \textwidth]{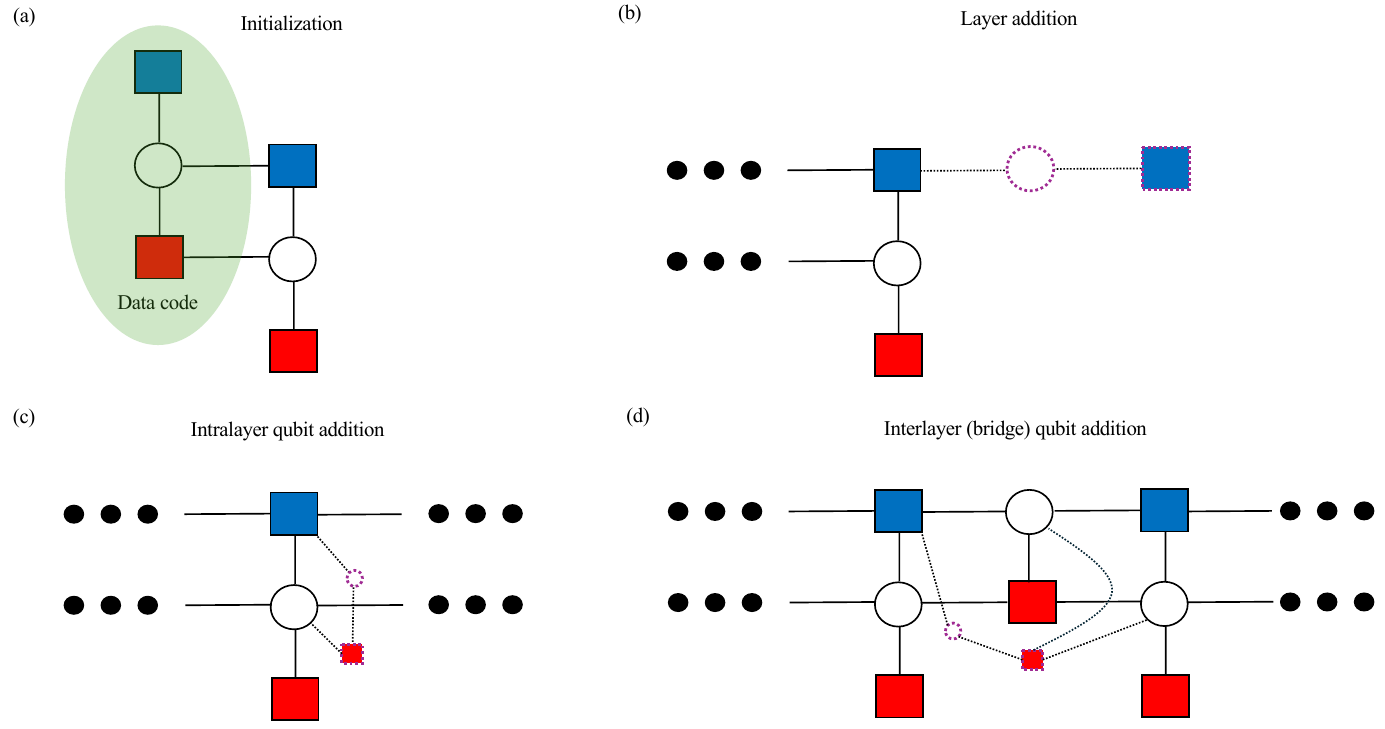}
     \caption{A demonstration of the random ancilla construction algorithm. (a) Initialization of ancilla system. Green shaded region represent the data code. (b) Layer addition. New layers are added at the end of the layered ancilla system. (c) Intralayer qubit addition. Qubits and Z-checks are added to boost expansion. (d) Interlayer qubit addition. Qubits and Z-checks are added, where Z-checks are connected to qubits in consecutive layers.}
    \label{fig:algo_demo}
\end{figure*}

In this section, we provide a detailed description of the randomized algorithm for ancilla construction. The inputs are given by the data code parity check matrices, $H_{X,Z}$, and the logicals to be measured, $L_X$. Here we assume the logicals to be measured are all of X-type. As mentioned in the main text, in simulations, we aim to minimally increase the parity check matrix. In low-rate gauging, it is commonly found a minimum of weight-2 increase in the check degree~\cite{yoder2025tourgrossmodularquantum,Cross2024ImprovedSurgery, Ide_2025}. Thus, we have chosen to adopt the simplest transversal-type connections in $\Gamma_{0,1}$, which has $\abs{\Gamma_{0,1}} = 1$, to gaurantee such properties.

One key motivation for the numerics is to demonstrate that even without the symmetry defined in Definition~\ref{def:gen_par_structure}, the idea of high-rate surgery brings significant savings in overhead. Therefore, we assume no symmetry in the logicals to measure. With such an assumption, our HGP ancilla construction would construct an HGP ancilla as the product of a $k=t$ classical code and a repetition code. Such a construction is not efficient in practice since for finite size the constant factors are highly nontrivial. Therefore, we consider a more flexible randomized construction, which is similar to the randomized edge addition algorithms in Refs.~\cite{Ide_2025,williamson2024lowoverheadfaulttolerantquantumcomputation}. Our major innovation is to adapt it to general check matrices, which can be understood as adding hyperedges. Such a generalization loses the asymptotic gaurantee of expander graphs, but it is observed to perform well in practice. Moreover, while the way we picked which hyperedges to add still has random components, we developed well-defined heuristics to boost its effectiveness.

\begin{algorithm}[t]
\caption{Randomized construction of distance-preserving ancilla system}\label{alg:rand_ancilla_construction}
\KwIn{ Data code parity checks $\partial^C_2, \partial^C_1$, logicals to measure $L_X$, degree limit $c$.}
\KwOut{Ancilla parity checks $\partial^A_1, \partial^A_0$, chain maps $\Gamma_{0,1}$.}
Initialize ancilla parity check as $\partial^A_1\leftarrow \begin{pmatrix}
        \partial^C_1\vert_{L_X} \\ \partial^{l}
\end{pmatrix}$, where $\partial^C_1$ is truncated to the union of $L_X$ and low-weight logical checks, $\partial^{l}$, are found through a greedy search. Construct generating basis for $\ker{\sbkt{\partial_2^A}^T}$ with weight reduction heuristics. Initialize $\partial^A_1$ by the rows of the generating basis that has degree less than $c$.\\
\uIf{$\abs{\partial^A_1}<c$}{Continue}
\Else{Return fail.}
% \EndIf
Initialize chain maps as $\Gamma_{0,1}\leftarrow \begin{pmatrix}
        I_{0,1}\\ \mathbf{0}
\end{pmatrix}$ are of size $\abs{C_{0,1}}\times \abs{A_{0,1}}$, respectively.\\
\tcp{$I_{0,1}$ are projections onto the support of $L_X$ and onto the adjacent checks in $\partial_1^C$ respectively.}
Initialize a list of layers, with the initial checks as the first layer with index 1. Each layer with index $i$ is parametrized by the layer's check matrices, $\partial_{0,1}^i$.\\
Initialize empty bridge checks list.\\
% \While{$d_{\text{merge}}< d$}
\While{$d_{\text{merge}}< d$}{
        Sample logically distinct low-weight X-type logicals, $l_X$ of weight $d_{\text{merge}}$.\\
        \For{$i\leftarrow 0$ \KwTo max layer index}{
                Sample low-weight row vectors through (1) generate low-weight combinations, $v$, to obtain $w = \sbkt{\partial_1^i}^T v$, such that $\abs{v}, \abs{w}\leq c$ or (2) copying checks from layers of index $i-1, i+1$, if exist.\\
                Compute the score of each low-weight check generated, $w$. The score is the number of $l_X$ that anticommutes with $w$.\\
                \uIf{exists $w$ with $\text{score}> 0$}{
                        Choose the $w$ with max score. Modify the layer's check matrix as $\partial^i_1 \leftarrow \begin{pmatrix}
                                \partial_1^i \\ w
                        \end{pmatrix}$. If $w$ is generated through check combination, update $\partial_0^i \to \begin{pmatrix}
                                \partial_0^i & \mathbf{0}\\
                                v^T & 1
                        \end{pmatrix}$. Otherwise, update the bridging checks.\\
                }
                \Else{
                        Continue to next loop\\
                }
        }
        \uIf{no checks have been added}{
                Let current last layer index be $m$. Add a new layer, $m+1$. Initialize the new layer's checks as null matrices. Add bridge checks that connects transversally the X-checks of $m, m+1$ through bridge qubits.\\
                \tcp{See Figure~\ref{fig:algo_demo} and discussion.}
        }
}
\Return $\partial^A_1, \partial^A_0, \Gamma_{0,1}$
\end{algorithm}

See Figure~\ref{fig:algo_demo} for an illustration of the algorithm's key steps and see Algorithm~\ref{alg:rand_ancilla_construction} for a complete description. The inputs to the algorithm is $\partial_{1,2}^C$ and a logical representation of $L_X$. We also set a check degree restriction, $c$, on the merge code degree. From discussion before, $c$ is set degree of data code plus 2 in simulations we presented in this work. On a high-level, the algorithm constructs a layered ancilla system. Broadly speaking, the algorithm can be divided into 4 steps, where one initializes the ancilla system to perform the right measurements (Figure~\ref{fig:algo_demo}(a)), adds qubits and Z-checks to increase expansion (Figure~\ref{fig:algo_demo}(c, d)), and adds qubits and X-checks to lower check degrees (Figure~\ref{fig:algo_demo}(b)). The second step is a generalization of the random add edge construction~\cite{yoder2025tourgrossmodularquantum,Cross2024ImprovedSurgery, Ide_2025}. The last step is crucial, especially for larger codes, and is performed when one cannot increase expansion while maintaining the degree limit of the parity check.

\paragraph*{\textbf{Initialization}} The ancilla system, $\partial_{0,1}^A$, and the chain maps are initialized. We start by initializing $\partial_1^A$ as $\partial_1^C\vert_{L_X}$, where the subscript denotes restricted to the support of $L_X$, and $\Gamma_{0,1}\leftarrow \begin{pmatrix}
        I_{0,1}\\ \mathbf{0}
\end{pmatrix}$. Here, $I_{0,1}$ are projections onto the support of $L_X$ and all adjacent checks respectively. Such initializations gaurantee the measurement of all logicals in $L_X$. However, to avoid measuring additional logicals, we add an additional check $\partial^l$ such that the full initial boundary is 
\begin{eqnarray}
                \partial^A_1= \begin{pmatrix}
\partial^C_1\vert_{L_X} \\ \partial^{l}
        \end{pmatrix}
\end{eqnarray}
In general, there is no theoretical gaurantee that $\partial^{l}$ is LDPC for general basis (see Methods for a detailed discussion). However, we find as a heuristic that as long as the physical circuit has circuit degree less than $c$, we can find a low-weight $\abs{\partial^{l}}$. We search for $\partial^{l}$ with an iterative local greedy algorithm: for each step, we find checks $\partial^l$ through Algorithm~\ref{alg:subsystem_log} by letting the inputs $H =L_X \Gamma_1^{-1}$, $H_c = \mathbf{0}$, and $L^c = \ker{\Gamma_1^{-1}}\setminus H$. Here, the inverse should be understood as the pseudoinverse, and $\partial^l$ is a random lowest-weight choice from the output, $\set{L}$. This step can be understood as finding low-weight checks to add to a classical code $\Gamma_1^{-1}$ such that a specified kernel space can be retained. This process is repeated until $\Gamma_1\sbkt{\ker{\Gamma_1^{-1}}}$ is in the span of $\begin{pmatrix}
        \sbkt{\partial_2^C}^T\\
        L_X
\end{pmatrix}$. Note that this step is the only procedure in the algorithm that can lead to failure in ancilla construction. Nevertheless, the failure is mostly due to a poor choice in logical basis or a high physical circuit degree, so there is little one can do from a algorithmic perspective. Lastly, we search for a low-weight generating basis for $\ker{\sbkt{\Gamma_1^{-1}}^T}$ and include all low-weight checks with degree less than $c$ to be the initial $\partial_0^A$. The ancilla system at this point constitutes layer with index 1.

\begin{algorithm}[t]
\caption{Search for low-weight subsystem logicals}\label{alg:subsystem_log}
\KwIn{Parity check $H$, conjugate parity check $H^c$, bare conjugate logical $L^c$, max logical weight $w_m$, number of trials $n_{\text{trials}}$.}
\KwOut{Set of low-weight dressed logicals $\set{L}$, dressed distance $d$.}
Initialize sets of dressed logicals, $\set{L}$, as empty set.\\
\For{$i\leftarrow 1$ \KwTo $n_{\text{trials}}$}{
        Randomly shuffle rows of $H$ to get $H_i$. \\
        Perform Gaussian elimination on $H_i$ to get a generating basis for its kernel, defined as a matrix $K$.\\
        Compute $M := K \sbkt{L^c}^T\mod 2$, and assume it is a $R\times C$ dimensional matrix.\\
        \For{$r\leftarrow 1$ \KwTo $R$}{
                \uIf{\text{Sum}$(M_{r, :})>0$ and \text{Sum}$(K_{r, :})<=w_m$}{
                        $\set{L}\leftarrow \set{L} + K_{r, :}$\\
                }
        }
}
Dress distance, $d\leftarrow \min\set{\abs{L}}$.\\
\Return $\set{L}, d$
\end{algorithm}

\paragraph*{\textbf{Increase expansion}} The expansion is increased through addition of ancilla qubits and $Z$ checks. This subroutine is a generalization of the previous graph algorithms. However, instead of a fully random choice of hyperedge, we make a local greedy choice: we start with computing the current dressed distance and sample a set of distinct low-weight logicals, $\set{L}$, with Algorithm~\ref{alg:subsystem_log} of the whole merged code. We then search for new qubits to add from layer-to-layer. Consider a layer with index $i$ and the intralayer checks are $\partial^i_{0,1}$, there are two types of qubits one can add.
\begin{enumerate}
        \item The new qubit is connected to ancilla X-checks as a low-weight combination of the X-checks in the current layer. In practice, we low-weight vectors $v$ and defined $w = \sbkt{\partial_1^i}^T v$. If we add such a qubit, the layer checks are modified as
        \begin{eqnarray}
                        \partial_1^i \to \begin{pmatrix}
                                \partial_1^i \\ w
                        \end{pmatrix}, \partial_0^i &\to& \begin{pmatrix}
                                \partial_0^i & \mathbf{0}\\
                                v^T & 1
                        \end{pmatrix}
        \end{eqnarray}
        so we filter by requiring the modified check matrices to have degree under $c$. 
        \item The new qubit is connected to current layer X-checks identically with the connection of an existing qubit in layer $i-1$ or $i+1$ to that layer's X checks. As we will see later and also in Figure.~\ref{fig:algo_demo}, the X-checks in each layer are of the same size, so to have identical connections is well-defined. Without loss of generality, consider duplicating a X-check in the $i-1$th layer with X-check $w$, we modify the check matrices as 
        \begin{eqnarray}
                        \partial_1^i &\to& \begin{pmatrix}
                                \partial_1^i \\ w
                        \end{pmatrix}\\
                        \begin{pmatrix}
                                 & \vdots&\vdots&\vdots &\\
                                \dots & \partial_0^{i-1} & \mathbf{0} & \mathbf{0}& \dots \\
                                \dots &I^{i-1} & \delta^{i-1,i} & I^i & \dots \\
                                \dots &\mathbf{0} & \mathbf{0} & \partial_0^{i}& \dots \\
                                & \underbrace{\vdots}_{\text{Layer }i-1}&\underbrace{\vdots}_{\text{Bridge}}&\underbrace{\vdots}_{\text{Layer }i} &
                        \end{pmatrix}&\to& \begin{pmatrix}
                                 & \vdots&\vdots& \vdots&\vdots &\\
                                \dots & \partial_0^{i-1} & \mathbf{0} & \mathbf{0}& \mathbf{0}& \dots \\
                                \dots &I^{i-1} & \delta^{i-1,i} & \mathbf{0} & I^i & \dots \\
                                \dots & e^T_w & w^T & e^T_w & \mathbf{0} & \dots\\
                                \dots &\mathbf{0} & \mathbf{0} & \mathbf{0}&\partial_0^{i}& \dots \\
                                & \vdots&\vdots&\vdots&\vdots &
                        \end{pmatrix}\label{eq:bridge_mat}
        \end{eqnarray}
        where in the second line we show the modification of relevant intralayer Z-checks and also bridge checks and qubits. $\delta^{i-1,i}$ describes the bridge checks, and its rows are formed by common subsets of columns of $\partial_0^{i-1,i}$. $e_w$ is unit vectors with nonzero entry at the corresponding duplicated qubit's position. One can see clearly from Eq.~\eqref{eq:bridge_mat} that the added bridge Z-checks will have weight $\abs{w} + 2$. Thus, overall we restrict $\abs{w}\leq c-2$ so that all qubits are available for copying.
\end{enumerate}
With either choice of qubits, suppose the Z-checks added is $Z_v$, we compute a score for each qubit addition as $s_v:= \text{Sum}\sbkt{Z_v^T L\mod 2}$. Intuitively, such a score represents how many low-weight logicals anticommutes with the candidate Z-check. By the stabilizer and logical update rule [CITE?], any logical that commutes with the added Z-checks will remain a logical in the enlarged system with the same weight. We iterate over existing layers from 1 to the end, and as long as we find a layer with nonzero $s_v$, we add the qubit with max $s_v$ and continue to the next iteration. Note that there is no gaurantee one can always find low-weight $v,w$ such that the score is nonzero. When one fails to do so, we enter the procedure of adding layers.

\paragraph*{\textbf{Reducing code degrees}} To have a mechanism of adding X-checks meaningfully is crucial for the success of our algorithm since notably the steps before are all adding qubits and Z-checks. If no X-checks can be added, the weight of the X-checks would soon exceed the code degree limit set. We choose to add X-checks in a conservative way to gaurantee that the algorithm succeeds, i.e. in the worst case we recover the HGP ancilla construction and still return an ancilla construction with sufficient distance. We add X-checks by adding a whole layer of X-checks and qubits, connected 1-to-1. Assume the current maximum number of layers is $m$, we add a new layer through
\begin{eqnarray}
                \begin{pmatrix}
                                 & \vdots\\
                                \dots & \partial_1^{m}
                \end{pmatrix} &\to& \begin{pmatrix}
                                 & \vdots & \vdots\\
                                \dots & \partial_1^{m} & I\\
                                \dots& \mathbf{0}& I
                \end{pmatrix}\\
                 \partial^A_0 &\to& \begin{pmatrix}
                        \partial^A_0 & \mathbf{0}
                \end{pmatrix}
\end{eqnarray}
This step cannot increase the distance or eliminate any low-weight logicals, but after adding such a layer, one is gauranteed to find low-weight qubits that can increase expansion in the previous procedure.

One of the subtlety we have skipped is that to gaurantee absolutely one can reach HGP ancilla construction as worst case, it is necessary to only consider case 1 in the increase expansion procedure when the current layer has a full kernl,
\begin{eqnarray}
                \ker{\partial_1^i} = \ker{\partial_1^1}
\end{eqnarray}
where $\ker{\partial_1^1}$ denotes the kernel space of the first layer. Such condition is required since we have allowed the algorithm to randomly add hyperedge combinations, which is not present in HGP constructions. If this condition is not present, it is possible for the algorithm to not find a possible nontrivial qubit choice even with added layers. However, we remark that this failure scenario only appeared when we simulated very large code examples, and it contributed the increased overhead per logicals measured. For smaller codes, one can often find much smaller distance-preserving ancilla constructions by not enforcing this condition. In addition, we only considered bridge qubits between adjacent layers, which is not necessary since locality is not a focus in this work. Therefore, we imagine that there are many ways to improve the algorithm by adding flexibility and clever algorithms. We leave such optimizations to future work.

In conclusion, we found that two points are crucial in the algorithm design. The first is to provide the ancilla system procedures to increase the number of $X$-checks. This is, in our opinion, an overlooked mechanism by prior works\cite{yoder2025tourgrossmodularquantum,Cross2024ImprovedSurgery, Ide_2025}. Their approaches worked both because they did not enforce a strict degree constraint and that they are working with small code examples measuring single logicals. However, to scale up while preserving sparsity, one would require a procedure similar to the layer addition procedure we outlined. Secondly, we found that to rank the possible qubit additions by their scores significantly reduces the system overhead. This holds even if we apply our algorithm to measure a single logical.

\section{Measurement of mixed-type logicals \label{app:mixed_meas}}

In the above, we have restricted to measurement of single Pauli type logicals in CSS codes for simplicity. In this section, we comment on how our approach generalizes to general logical measurements. Suppose we consider CSS codes and measurement of mixed types of logicals. One can generally write the logical measurement set as 
\begin{eqnarray}
                M = \set{X^L_i Z^L_j \vert i\in K_X, j\in K_Z}
\end{eqnarray}
where $K_{X,Z}$ denotes the set of logical X and Z indices that are in $M$. Here, we neglect the $\pm 1$ factors for simplicity. The discussion is relatively simple when $K_X\cap K_Z = \varnothing$: without lost of generality, one can first construct the ancilla system for measuring all $M_X = \set{X^L_i\vert i\in K_X}$ and then for the Z-logicals in $K_Z$. To properly measure the products instead of measuring them individually, one add logical checks following the procedure described in Methods and in Appendix~\ref{app:random_algorithm}. Similar ideas for low-rate surgery has been known as bridge qubits~\cite{Cross2024ImprovedSurgery,yoder2025tourgrossmodularquantum} or adapters~\cite{swaroop2025universaladaptersquantumldpc}. 

It is less straightforward when there are overlap in $K_{X,Z}$. The above construction works is that the construction of each X and Z ancilla system measure only operators that commutes with the current stabilizer group, which aligns with the surgery framework we have developed. Nevertheless, when there is overlap, such an assumption no longer holds. One remedy is to use twist-free surgery~\cite{PRXQuantum.3.010331,Cowtan2025ParallelSurgery,Guo2025TimeEfficient}, where the measurement circuit is decomposed into two parts, each measuring a commuting joint Pauli product. The drawback is that one would need a resource state, $\ket{Y}$. This overhead is acceptable in the asymptotic regime since the Y resource state is unchanged after the measurement and can be used repeatedly. One way of preparing these Y resource states would be to perform low-rate surgery one-by-one. However, it is unknown how to prepare these Y resource states in parallel. Thus, more work is needed to assess the overhead associated with such resource states.

In Ref.~\cite{Cowtan2025ParallelSurgery}, the canonical logical basis~\cite{Cowtan2025ParallelSurgery, gottesman1997stabilizercodesquantumerror} was proposed to circumvent this challenge. In such a basis, for any physical qubit in $C_1$, it is either X-support or Z-support of any logicals, but not both. This enables the scheme in Ref.~\cite{Cowtan2025ParallelSurgery} to prepare the Y states simulatenously. We refer interested readers to Ref.~\cite{Cowtan2025ParallelSurgery} for a more detailed discussion. While our scheme can directly adopt a similar solution, there is also an alternative approach: from our discussion in Appendix~\ref{app:general_framework} and Lemma~\ref{lem:surgery_measurements_framework}, we remarked that the key to the logical measurement is through redundancies in the ancilla X-checks. In fact, this idea is not limited to measuring a single type of logicals: in the symplectic form, consider the chain maps $\Gamma_{0,1}$ as
\begin{eqnarray}
                \Gamma_{\alpha} = \begin{pmatrix}
                        \Gamma^X_{\alpha} & \Gamma^Z_{\alpha}
                \end{pmatrix}
\end{eqnarray}
with $\alpha = 0,1$. $\Gamma_{\alpha}$ is a matrix of size $\abs{A_1}\times 2\abs{C_1}$. $\Gamma_{\alpha}$ allows the ancilla $A_1$ checks to connect to the data qubits, $C_1$, through either $X, Y, Z$ connections. One the chain maps form a commuting diagram, the logical action is given by $\Gamma_1 \sbkt{\ker{\partial_1^A}}$ as in Lemma~\ref{lem:surgery_measurements_framework}. However, while it is not so different from the perspective of design the kernel, $\ker{\partial_1^A}$, there is generally no gaurantee that all checks in $\Gamma_{\alpha}$ will commute with itself, i.e. $\Gamma_{\alpha} \Omega \Gamma_{\alpha}^T = 0$, with $\Omega$ being the symplectic matrix. For example, the logical set, $M$, can contain pairs of commuting logicals that have anticommuting supports on even number of qubits. Nevertheless, this will not occur if the canonical basis is adopted: $\Gamma_1$ will either not connect to a data qubit or connect to it through a single type of Pauli operator. One should note that the canonical basis is not limited to CSS codes and so is our discussion above. The major change would be that when one construct the ancilla system through, for example, taking union of the conjugate checks (see discussion in Appendix~\ref{app:random_algorithm}), one should consider all checks that anticommute with all individual parity checks added through the ancilla system.
 
\section{Constructions of spatially coupled codes \label{sec:app_SC_codes}}
We follow the standard construction of Ref.~\cite{Yang2025SCQLDPC} and identify the instances listed in Table~\ref{tab:SC_code_params} via numerical search. Here, we briefly sketch the construction and refer the reader to Ref.~\cite{Yang2025SCQLDPC} for full details. A $[[n, k, d]]$ spatially coupled hypergraph product (SC-HGP) code is built by starting from an $L \times L$ array of identical base HGP codes — each obtained as the hypergraph product of two identical classical codes with check matrix $H_C \in \mathbb{F}_2^{r_C \times n_C}$ — and introducing couplings between them. Each base HGP code has parameters $[[n_0 = r_C^2 + n_C^2,\ k_0 \ge n_0 - 2r_C n_C,\ d_0]]$, leading to a SC-HGP code with parameters $n = n_0 L^2$ and $k \ge k_0 L^2$. As an example, for the $[[1125, 245, 10]]$, we choose $r_C = 3$, $n_C = 6$, and $L = 5$.

\end{document}